\documentclass[11pt, onecolumn, draftcls]{IEEEtran}

\IEEEoverridecommandlockouts 
\usepackage{xcolor}
\usepackage{float}
\usepackage{graphicx}
\usepackage{subfigure}
\usepackage{color}
\usepackage{epsfig}
\usepackage{amssymb}
\usepackage{amsmath}
\usepackage{amsthm}
\usepackage{latexsym}
\usepackage{setspace}
\usepackage{bbm}
\usepackage{dsfont}
\usepackage{xspace}
\usepackage[ruled,vlined]{algorithm2e}
\usepackage[normalem]{ulem}

\usepackage[symbol]{footmisc}

\newcommand{\bPr}[1]{{\mathbb{P}}\left(#1\right)}
\newcommand{\bPPr}[2]{{\mathbb{P}}_{#1}\left(#2\right)}

\newcommand{\bP}[2]{\mathrm{P}_{#1}\left({#2}\right)}

\newcommand{\bE}[2]{{\mathbb{E}}_{#1}\left\{{#2}\right\}}
\newcommand{\bEE}[1]{{\mathbb{E}}\left[{#1}\right]}

\newcommand{\cA}{{\mathcal A}}
\newcommand{\cB}{{\mathcal B}}

\newcommand{\cG}{{\mathcal G}}
\newcommand{\cH}{{\mathcal H}}

\newcommand{\cN}{{\mathcal N}}
\newcommand{\mN}{{\mathbb N}}

\newcommand{\cP}{{\mathcal P}}

\newcommand{\mR}{{\mathbb R}}

\newcommand{\cX}{{\mathcal X}}
\newcommand{\cY}{{\mathcal Y}}

\newcommand{\by}{\mathbf{y}}

\newcommand{\ep}{\varepsilon}

\usepackage{color}

\newtheorem{theorem}{Theorem}

\newtheorem{corollary}[theorem]{Corollary}
\newtheorem*{corollary*}{Corollary}

\newtheorem{lemma}[theorem]{Lemma}
\newtheorem*{lemma*}{Lemmas}

\theoremstyle{remark}
\newtheorem*{remark*}{Remark}
\newtheorem*{remarks*}{Remarks}

\theoremstyle{definition}
\newtheorem{definition}{Definition}

\newcommand{\indicator}{{\mathds{1}}}

\allowdisplaybreaks

\title {Communication Complexity of Distributed High Dimensional Correlation
  Testing}

\author{{K.~R.~Sahasranand}\,\, \and\,\,
  {Himanshu Tyagi}}

\date{} 

\begin{document}
\maketitle

{\renewcommand{\thefootnote}{}\footnotetext{
\noindent The authors are with the Department of Electrical Communication Engineering,
Indian Institute of Science, Bangalore 560012, India.  Email:
\{sahasranand, htyagi\}@iisc.ac.in

  A preliminary  version of this paper~\cite{krs-ht-isit18} was presented at the IEEE International
  Symposium on Information Theory, Vail, USA, 2018.
  }
}
\renewcommand{\thefootnote}{\arabic{footnote}}
\setcounter{footnote}{0}

\begin{abstract}
Two parties observe independent copies of a $d$-dimensional vector and
a scalar. They seek to test if their data is correlated or not, namely
they seek to test if the norm $\|\rho\|_2$
of the correlation vector $\rho$ between their observations exceeds
  $\tau$ or is it $0$. To that end, they communicate interactively and declare the
output of the test. We show that roughly order $d/\tau^2$ bits of
communication are sufficient and necessary for resolving the
distributed correlation testing problem above. Furthermore, we establish
a lower bound of roughly $d^2/\tau^2$ bits for communication needed for
distributed correlation estimation, rendering the estimate-and-test approach
suboptimal in communication required for distributed correlation testing. For the one-dimensional case with
one-way communication, our bounds are tight even in the constant and
provide a precise dependence of communication complexity on the probabilities of error of two types. 
\end{abstract}

\section{Introduction}
Parties $\cP_1$ and $\cP_2$ observe jointly Gaussian random variables $X^n$ and $Y^n$, respectively, comprising independent and identically distributed (i.i.d.) samples $(X_t, Y_t)$, $1\leq t \leq n$, with $X_t\in \mR^d$,
$Y_t\in \mR$, and such that $\bEE{Y_1\mid X_1}=\rho^TX_1$. They
communicate with each other to determine if their observations are
correlated, i.e., to test if $\|\rho\|_2\geq \tau$ or $\|\rho\|_2=0$. 
For a given probability of error requirement and an arbitrary large $n$, what is the minimum
communication needed between the parties?

Note that we have chosen the distribution to be Gaussian just for convenience. Since we allow the number of samples
to be arbitrarily large, even when $X$ and $Y$ are not Gaussian, we can replace subset of samples with their sample
means and use the central limit theorem (Berry-Esseen approximation) to do similar calculations as those presented in this paper.
Indeed, all the results of this paper extend to the case when $X_t$ and $Y_t$ are distributed uniformly over
$\{-1,1\}^d$ and  $\{-1,1\}$, respectively, and  $\mathbb{E}[Y_1|X_1] = \rho^T X_1$.
In
another direction, it is seemingly restrictive to consider the conditional expectation to be a linear function of the observation vector $X_1$.
However, as in nonparametric regression, we can express $E[Y_1|X_1=x_1]$ over an orthonormal basis for the $L_2$ space and replace the coordinates of $x_1$
with the evaluation of $f_i(x_1)$ for basis functions $f_i$. The dimension $d$ is chosen sufficiently large to capture most of the energy.
In summary, our seemingly restrictive setup can be easily enhanced to the non-Gaussian setting with nonlinear conditional expectation by using appropriate kernels (basis functions), as in nonparametric regression. In fact, this is the reason we call this problem correlation testing, instead of just independence testing.

This problem is an instance of a distributed hypothesis testing
problem, which has been studied for several decades in the information
theory literature starting with the seminal work~\cite{AhlCsi86} and
closely followed by~\cite{Han87}.
Most formulations in this literature focus on the tradeoff between the
error exponent and communication rate per sample for simple binary hypothesis testing
problems; see~\cite{HanAmari98} for a survey. We remark that our
setting differs from these classic settings since we 
consider a
composite hypothesis testing problem. Furthermore, we do not focus on the
error-exponent and allow arbitrarily large number of samples $n$. 
In particular, the error exponent can
be shown to be $0$ when we restrict to communication of rate $0$
($cf.$~\cite{ShalabyPapamarcou92}), which is an allowed regime for us
since we can take as many samples as we like to minimize the communication.

The problem of distributed independence testing with multiple rounds of interactive
communication was studied in~\cite{XiangKimAllerton12,
  XiangKimISIT13}. Similar problems with more general hypotheses
  or  more elaborate communication
models have been considered in~\cite{ZhaoLai14,
  WiggerTimo16, katz2016collaborative, sreekumar2019distributed}. Error exponent for the conditional independence
testing problem is studied in~\cite{rahman2012optimality}, where both
upper and lower bound for it are obtained. Recently, and subsequent to the publication of the initial version~\cite{krs-ht-isit18} of
this paper, related problems were considered
in various works. In~\cite{hadar2019error}, an improved upper bound on
the Stein exponent for testing between two known positive Gaussian
correlations is provided.
The communication complexity of estimating one-dimensional Gaussian
correlations was established in~\cite{hadar2019communication}
 and that of independence testing over discrete alphabet in the large
sample regime was characterized in~\cite{andoni2019}. The tradeoff between communication complexity and
sample complexity for detecting pairwise correlations is studied in
\cite{dagan2018detecting}. A related line of recent work considers composite hypothesis testing under communication, privacy, and shared randomness constraints \cite{acharya2019inferenceI, acharya2019inferenceII,  pmlr-v89-acharya19b, acharya2019domain}. However, the constraints are placed on each independent sample rather than on parties observing multiple correlated samples. 
In particular, none of the prior works consider our specific composite hypothesis testing
  problem.

The related problem of estimating the correlation vector for the same setting as ours was
studied in~\cite{hadar2019distributed}. It is plausible to use the distributed estimation
  scheme of~\cite{hadar2019distributed} or similar estimation schemes
  to do testing, but we establish a lower bound to show that this
  approach will be suboptimal in communication requirement. 
Our main result is the characterization of the minimum communication
needed for distributed correlation testing. Our proposed distributed
test uses one-way communication and solves the $d$-dimensional
problem by reducing it to the case $d=1$. This is done by multiplying
the observation vectors $X_t$s of $\cP_1$ with a scaled random sign
vector. Specifically, for $d=1$, our test entails the use of shared randomness to sample a vector that is
close to $\cP_1$'s overall observation $(X_1, ..., X_n)$, sending the
identity of this vector to $\cP_2$, and then $\cP_2$ checking if its
observation vector $(Y_1, ..., Y_n)$ is close to this vector as
well. We show that this test requires roughly $\max\{(1/\tau^2)\log
1/\ep, (1/\tau^2 -1)\log 1/\delta\}$ bits of communication 
to get probabilities of false alarm and missed detection to be less
than $\delta$ and $\ep$, respectively, when $n$ is sufficiently
large. For a general $d$, noting that the multiplication with random sign will yield a one dimensional
correlation testing problem with correlation roughly
$\|\rho\|_2/\sqrt{d}$, we show that the $d$-dimensional problem can be
resolved using roughly $(d/\tau^2)\max\{\log 1/\ep, \log 1/\delta\}$
bits of communication.

Our proposed test is practical. In fact, we have
  simulated a version with slightly different parameters than those
  presented in our theoretical analysis below; the empirical
  performance is depicted in Figure~\ref{fig:commVsErr}. A phase
  transition in probability of error can be seen clearly when we
  communicate number of bits proportional\footnote{As
    will be seen below, our proposed test uses a ``median trick'' to convert the
  one-dimensional test to a $d$-dimensional test. In our simulation,
  even the probabilities of correctness for the one-dimensional test
  are boosted to the desired levels by repeating the tests and using a
  similar ``median trick''.} to $d/\tau^2$.

\begin{figure}[h]
\centering
\includegraphics[scale=0.7]{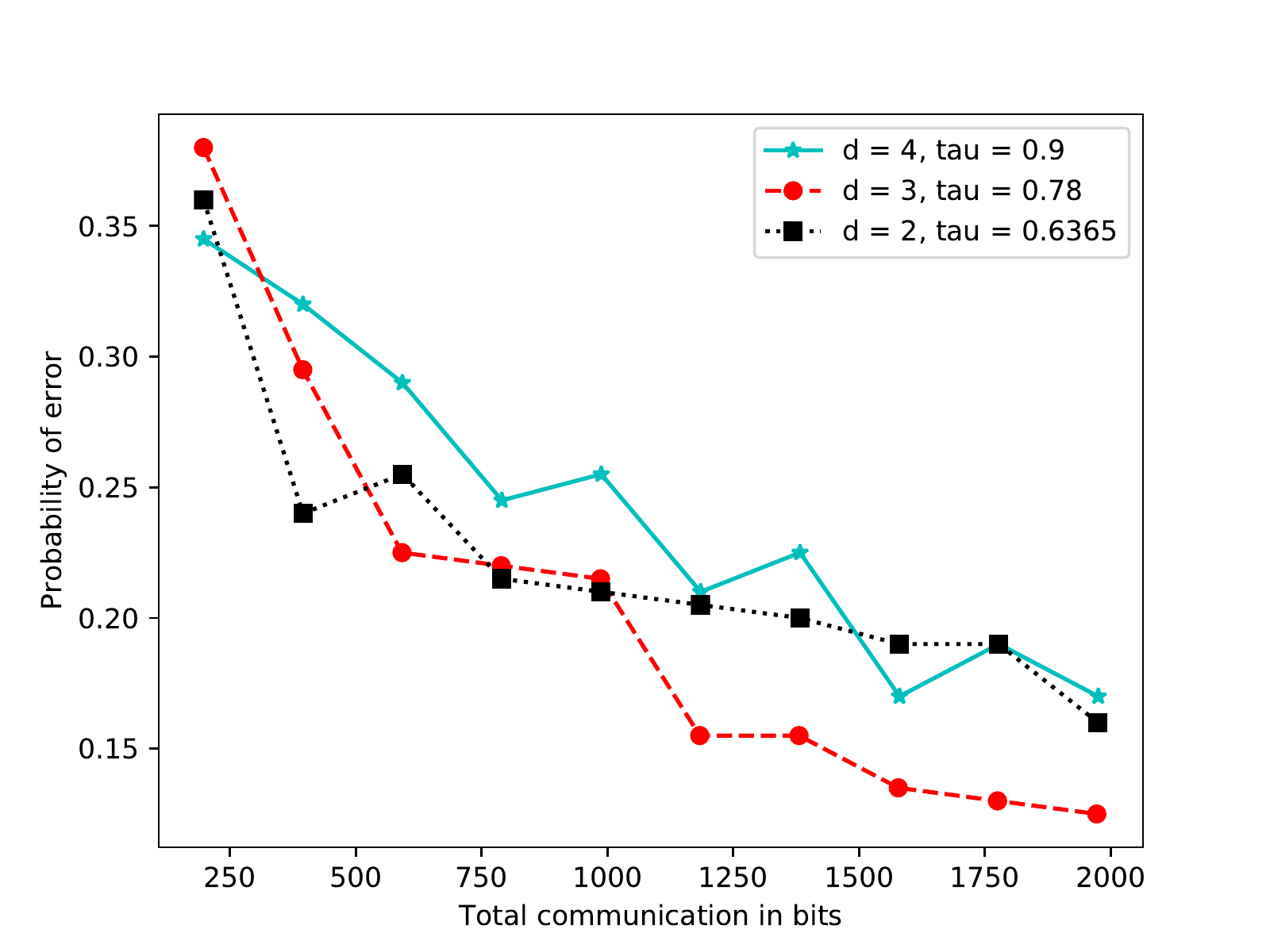}
\caption{Performance of the $d$-dimensional test for different values
  of $d$ and $\tau^2$ with $d/\tau^2 \approx 4.93$. The probability of
  error (in $y$-axis) is the average of probability of false alarm and
  probability of missed detection, evaluated by averaging over $100$ iterations.}\label{fig:commVsErr} 
\end{figure}

Interestingly, we establish a lower bound that shows that the amount
of one-way 
communication used by our protocol for $d=1$ is optimal among all
one-way communication protocols. We show this bound by using the notions
of hypercontractivity and reverse hypercontractivity
($cf.$~\cite{Bonami70, Gross75,
  Beckner75, Borell82, AhlswedeGacs76, Mossel13}).
Specifically, we note that the acceptance region corresponding to
one-way communication corresponds to a union of disjoint rectangle
sets. We use hypercontractivity to relate the
measures of rectangle sets
under the joint distribution corresponding to
$|\rho|>\tau$ and the product distribution corresponding to
$\rho=0$, which in turn leads to the required lower bound. In fact, 
by using the tensorization property we extend the bound to
a general $d$ to obtain a lower bound on one-way communication of
roughly $(d/\tau^2)\max\{\log 1/\ep, \log 1/\delta\}$. 

Recently, a strong data processing inequality for interactive correlation estimation was derived in ~\cite{hadar2019communication}. We invoke this result to show
that roughly $d/\tau^2$ bits of communication are needed even when interactive
communication is allowed, rendering our proposed one-way communication
protocol optimal among interactive protocols. We note that this bound is slightly weaker for
one-way communication than the one obtained using
hypercontractivity.

As mentioned earlier, the related problem of correlation vector
estimation was studied in~\cite{hadar2019distributed}. In that work,
an estimation protocol was given that uses roughly $d^2/\tau^2$ bits
of communication to estimate $\rho$ within a mean squared error of
$\tau$. Clearly, directly using this estimate to test will not be
communication optimal. However, a natural question arises: can we find
a better distributed estimation protocol that will remain
communication-optimal even for testing? We show that, in fact,
$d^2/\tau^2$ bits of communication are necessary for estimation,
whereby estimate-and-test strategy is strictly suboptimal for testing.

We note that our proposed one-way communication scheme is related
closely to the scheme in~\cite{GVJR2016} where communication for common randomness
generation ($cf.$~\cite{AhlCsi98}) was considered. 
We draw on the heuristic connection between independence testing and common randomness
generation highlighted in \cite{TyaWat14, TyaWat14ii} to adapt the
scheme of \cite{GVJR2016} to devise a distributed correlation test.

We remark that a simple scheme is possible for $d=1$ that
quantizes each value $X_t$ to its sign $\indicator\{X_t\geq 0\}$ and
uses the known sample complexity results for independence testing for
the collocated case ($cf.$~\cite{AcharyaDasKamath15}). This, too, will
result in a scheme that requires $O(1/\tau^2)$ bits of
communication. However, we noted in~\cite{krs-ht-isit18}, where we study the communication complexity of one-dimensional independence testing, that our
proposed scheme uses communication that is a constant factor lower that
this baseline scheme. On the other hand, our scheme requires
a much larger $n$ than this baseline scheme for $d=1$.

The remainder of the paper is organized as follows. We present our
problem formulation in the next section, followed by the main results
in Section~\ref{s:results}. Our distributed correlation test as well
as its analysis is presented in Section~\ref{s:scheme}. The proof of
our lower bounds for one-way communication is in
Section~\ref{s:proofLB}
and for interactive communication is in Section~\ref{s:proofLB2}. We
  conclude with a discussion and some extensions of our results in the final section.

{\em Notation.} Random variables are denoted by capital letters such
as $X$, $Y$, $etc.$; their specific realizations by the corresponding
small letters such as $x$, $y$, $etc.$; and their ranges by the
corresponding calligraphic forms such as $\cX$, $\cY$, $etc.$. $[N]$ denotes the set of integers $\{1,2,\ldots,N\}$. For a distribution $\mathbb{P}_\rho$ parametrized by $\rho$, we use $\mathbb{E}_\rho[X]$ to denote the expectation of the random variable $X$ with respect to $\mathbb{P}_\rho$. $\mathbb{P}_{\cH_i}(A)$ denotes the probability of event $A$ under hypothesis $\cH_i$. All the logarithms are to the base $2$; when
needed, we use $\ln a$ to denote the natural logarithm of $a$. For a
vector $a$, $a(i)$ denotes the $i$-th coordinate,
$a^T$ denotes the transpose, and
$\|a\|_p:=\left(\sum_{i=1}^d|a(i)|^p\right)^{\frac 1
    p}$ denotes the $\ell_p$-norm.

\section{Problem setup}
We consider jointly Gaussian random variables $X\in \mathbb{R}^d$ and
$Y\in\mathbb{R}$ with joint distribution as follows: for $\rho(i)
\in [-1,1], 1 \le i \le d$, we assume that
\begin{align}
\mathbb{E}[X(i)] &= 0, \quad \mathbb{E}[X(i)X(j)] = \indicator\left\{i=j\right\}, 1 \le
i\le j\le d, \nonumber \\ \mathbb{E}[Y|X] &= \sum_{i=1}^d \rho(i)
X(i),\quad \mathbb{E}[Y^2] = 1.
\label{e:model}
\end{align}
Note that the assumptions above imply $\mathbb{E}[Y] = 0$. Since
we assume $\bEE{Y^2}=1$, Jensen's inequality gives
\[
\|\rho\|_2^2 = \mathbb{E}\left[\mathbb{E}[Y|X]^2\right] \le \bEE{Y^2}
= 1.
\]
Alternatively, we can describe the joint distribution of $X$ and $Y$
as follows:
\[
Y = \rho^TX + \sqrt{1-\|\rho\|_2^2}Z,
\]
where $Z$ is a standard normal random variable, and $X$ and $Z$ are
independent.

Let $(X_t, Y_t)_{t=1}^{n}$ denote $n$ independent copies of $(X,Y)$.
We consider a distributed hypothesis testing problem where parties
$\cP_1$ and $\cP_2$ observe $X^n=(X_1, ..., X_n)$ and $Y^n=(Y_1, ...,
Y_n)$, respectively, and seek to resolve the following composite
hypothesis testing problem:
\begin{align*}
&\mathcal{H}_0^d : \|\rho\|_2 \ge \tau,\\ &\mathcal{H}_1^d : \rho = 0,
\end{align*}
where $\tau$ takes values in $(0,1]$ and is known to both the parties.
  
  To determine the true hypothesis, the parties communicate with
  each other interactively in multiple rounds. Specifically, the
  parties use an $r$-round interactive communication protocol
  $\pi$ that comprises mappings $f_1, ..., f_r$; $\cP_1$ and $\cP_2$
  use mappings $f_i$, $1\leq i \leq r$, to communicate in odd and even
  rounds $i$, respectively. Each mapping $f_i$ takes as input the
  local observation of the party, the previously seen communication,
  and a shared random variable $V$ available to both the parties and
  outputs a binary string. Formally, denoting by $C_j$ the random
  binary string sent in round $j$, we have
  \begin{align*}
    f_i: (X^n, C_1, ..., C_{i-1}, V)\mapsto C_i\in
    \{0,1\}^{\ell_i},\quad 1\leq i\leq r, i \text{ odd}, \\ f_i: (Y^n,
    C_1, ..., C_{i-1}, V)\mapsto C_i\in \{0,1\}^{\ell_i},\quad 1\leq
    i\leq r, i \text{ even},
  \end{align*}
  where $\ell_i$, $1\leq i \leq r$, denotes the length of
  communication in round $i$. The overall random communication $(C_1,
  ..., C_r)$ is called the transcript of the protocol and is denoted
  by $\Pi$. Furthermore, we denote by $|\pi|$ the length
  $\sum_{i=1}^r\ell_i$ of the transcript of the protocol. For
  simplicity, we describe our formulation below only for odd $r$; the
  case of even $r$ can be handled similarly.

For an odd $r$, an {\em $r$-interactive distributed test} $T=(\pi, g)$
consists of an $r$-round interactive communication protocol $\pi$ and
a decision mapping $g:(Y^n,\Pi, V)\mapsto \{0,1\}$. A {distributed
  test} $T=(\pi,g)$ constitutes an $(\ell,\delta,\ep,\tau)$-test with
observation length $n$ if $|\pi|=\ell$ and
\begin{align*}
&\bPPr{\cH_0^d}{g\big(Y^n, \Pi, V\big) = 1} \leq \delta, \text{ and}
  \\ &\bPPr{\cH_1^d}{g\big(Y^n, \Pi, V\big) = 0} \leq \ep.
\end{align*}
Our goal is to design a distributed test that communicates as few bits
as possible, while possessing the desired probabilities of error.
Formally, we seek bounds for the minimum communication for
$d$-dimensional correlation testing, defined next.
\begin{definition}
Given $\delta, \ep \in [0,1]$ and $\tau \in (0,1]$, the minimum
  $r$-round communication for $d$-dimensional correlation testing
  $C_d^r(\delta, \ep,\tau)$ is the least $\ell$ such that there exists
  an $(\ell,\delta, \ep,\tau)$-test $T=(\pi, g)$ with an $r$-round
  interactive communication protocol $\pi$, for all observations of
  length $n$ sufficiently large.

The minimum communication for $d$-dimensional correlation testing
$C_d(\delta, \ep, \tau)$ is the infimum over $r\in\mN$ of
$C_d^r(\delta, \ep, \tau)$.
\end{definition}
While we have formulated the problem for general $r$, our main focus
in this work is the minimum communication $C_d^1(\delta, \ep, \tau)$
for one-way communication protocols.  We characterize the dependence
of $C_d^1(\delta, \ep,\tau)$ on $\ep$ (respectively $\delta$), up to
absolute multiplicative constants and additive constants that may
depend on $\delta$ (respectively $\ep$). Furthermore, we show that the
dependence on $\ep$ is optimal up to constant factors, even when
additional rounds of interaction are available. We summarize our
results formally in the next section.  But before that we formulate
the related problem of correlation estimation.

We consider the problem of estimating $\rho$ for the joint
distribution given in~\eqref{e:model}. The observation of the parties
and the $r$-round interactive communication protocol is defined as
before; as above, we define the problem only for odd $r$. An {\em
  $r$-interactive distributed estimate} is a pair $(\pi,
\widehat{\rho})$ where $\pi$ is an $r$-round interactive communication
protocol and $\widehat{\rho}: (Y^n, \Pi, V)\mapsto
  \widehat{\rho}(Y^n,\Pi, V)\in [-1,1]^d$.

An $r$-interactive distributed estimate
$(\pi,\widehat{\rho})$ constitutes an $(\ell, \tau)$-estimate if
$|\pi|\leq \ell$ and
\begin{equation}
\bE{\rho}{\|\widehat{\rho}(Y^n, \Pi, V) - \rho\|_2^2}\leq
\tau^2,
\label{eqn:estaccuracy}
\end{equation}
where $\mathbb{E}_{\rho}$ denotes the expectation with respect to the distribution in~\eqref{e:model}.
\begin{definition}
Given $\tau \in (0,1]$, the minimum $r$-round communication for
  $d$-dimensional correlation estimation $\widetilde{C}_d^r(\tau)$ is the
  least $\ell$ such that there exists an $(\ell,\tau)$-estimate
  $T=(\pi, \widehat{\rho})$ with an $r$-round interactive
  communication protocol $\pi$, for all observations of length $n$
  sufficiently large.

The minimum communication for $d$-dimensional correlation estimation
$\widetilde{C}_d(\tau)$ is the infimum over $r\in\mN$ of
$\widetilde{C}_d^r(\tau)$.
\end{definition}
In the next section, we will provide a lower bound for
$\widetilde{C}_d(\tau)$, which establishes roughly that correlation
estimation requires much more communication than correlation testing.
\section{Main results}\label{s:results}
We have divided our results into three parts: upper bounds for
$C_d^1(\delta, \ep,\tau)$ achieved by our proposed scheme, a lower
bound for $C_d^1(\delta, \ep,\tau)$, and a lower bound for $C_d(\delta, \ep,\tau)$ with $r>1$. These parts are presented in separate
sections below. The upshot of our results is that our protocol with
$r=1$ uses almost minimum communication not only among one-way
communication protocols, but also among interactive
protocols. Furthermore, we establish a lower bound for
the correlation estimation protocol which shows that it requires
strictly more communication than correlation testing.

\subsection{Upper bounds for $C_d^1(\delta, \ep,\tau)$} 
Our goal in this work is to handle high dimensional correlation
testing. Interestingly, we establish a reduction which relates the
high dimensional case to $d=1$ case. To state our general result,
first we state the result for $d=1$.

\begin{theorem}\label{t:upper_bound1}   For every $\delta, \ep \in (0,1)$,
\begin{align*}
C_1^1(\delta, \ep, \tau) &\leq \frac{1}{\tau^2} \left(\sqrt{\log
  \frac{1}{\varepsilon}} +\sqrt{(1-\tau^2)\log
  \frac{1}{\delta}}\right)^2 \\ &\qquad+ \ln \left(\frac{2}{\tau^2}
\left(\sqrt{\ln \frac{1}{\varepsilon}} +\sqrt{(1-\tau^2)\ln
  \frac{1}{\delta}}\right)^2 + 1\right) + O\left(\sqrt{\log
  \frac{1}{\delta}\log\frac{1}{\varepsilon}}\right).
\end{align*}
\end{theorem}
To extend this result to the case of general $d$, we convert the
$d$-dimensional problem to the one-dimensional problem as follows:
Party $\cP_1$ applies a random rotation (using common randomness $V$)
to the observed vector $X$ to obtain $\widetilde{X}$. We show that the
first coordinate $\widetilde{X}(1)$ of the resulting vector
$\widetilde{X}$ and $Y$ have correlation coefficient roughly
$(\tau/\sqrt{d})$ under $\mathcal{H}_0^1$ (with high probability) and
correlation coefficient $0$ under $\mathcal{H}_1^1$. Using this fact
(and a reduction result provided in the next section), we get the
following upper bound for $C_d^1(\delta,\ep,\tau)$.

\begin{theorem}\label{t:upper_boundd} There exists a positive constant
  $c>0$ such that for every $\delta,\ep \in (0,1)$ we have
\begin{align*}
&C_d^1(\delta, \ep,\tau) \leq c \cdot \frac{d}{\tau^2} \cdot
  \max\left\{\log \frac{1}{\delta},\log \frac{1}{\varepsilon}\right\}
  + O\left(\ln \frac{d}{\tau^2}\right),
\end{align*}
where the second term has no dependence on $\delta$ or $\ep$.
\end{theorem}

\subsection{Lower bounds for $C_d^1(\delta,
  \ep,\tau)$}\label{ss:lowerbound} Our lower bound for the case $d=1$
matches the upper bound of Theorem~\ref{t:upper_bound1} up to additive
terms of lower order to yield the following characterization for
$C_1^1(\delta, \ep, \tau)$.
\begin{theorem}\label{t:SS1} 
For $\delta\leq 1/2$ and $\ep$ such that $\delta +
\ep^{\frac{1-\tau}{1+\tau}} \leq 1$, we have\footnote{With an abuse
  of notation, the $O(x)$ notation for the additive error denotes that
  the upper and lower bounds differ by at most an $O(x)$ term.}
\[
C_1^1(\delta, \ep, \tau) = \frac 1{\tau^2} \log \frac 1 \ep+
O_\delta\left(\sqrt{ \log \frac 1 \ep}\right),
\]
and for $\delta,\ep\in(0,1/2)$, we have
\[
C_1^1(\delta, \ep, \tau) = \frac {1-\tau^2}{\tau^2} \log \frac 1
\delta+ O_\ep\left(\sqrt{\log \frac 1 \delta}\right),
\]
where the notation $O_x$ denotes that the constant implied by $O$
depends on $x$.
\end{theorem}
The proof of this result uses the notions of {\em hypercontractivity}
and {\em reverse hypercontractivity} and is given in
Section~\ref{s:proofLB}.

In fact, we can relate the $d$-dimensional problem to the one-dimensional problem by revealing extra information to $\cP_2$ to obtain a
matching lower bound for Theorem~\ref{t:upper_boundd}, from which the next result follows.
\begin{theorem}
\label{t:ddim}
For $0 < \tau \le 1$, $\delta\in(0,1/2)$, and $\ep$ such that
$\delta+ \ep^{\frac{1-\tau}{1+\tau}} \leq 1$, we have
\[
C_d^1(\delta, \ep, \tau) = \Theta\left(\frac d{\tau^2} \cdot
\max\left\{\log \frac 1 \ep, \log \frac 1 \delta\right\} \right).
\]
\end{theorem}
We remark that the reduction of the general $d\geq 1$ case to the one-dimensional case used in the proof of lower bound differs from the
reduction in the upper bound; we provide the proof in
Section~\ref{s:proofLB}.  Nevertheless, it is interesting that we
obtain tight results by relating the high dimensional setting to the
one-dimensional setting.

\subsection{Lower bounds for $r\geq 1$}
Our final set of results provide lower bounds even for the interactive
setting, establishing the optimality of our proposed distributed test
even among interactive tests. To derive this lower bound, we use a
data processing inequality from~\cite{LiuCuffVerdu17}, which was used
in a similar context in~\cite{hadar2019communication}. In fact, using
this technique we can even derive a lower bound for the high
dimensional correlation estimation problem, showing that this problem
requires orderwise higher communication in comparison to correlation
testing.

We begin with the result for the correlation testing problem. Note
that we only prove optimality in the dependence on $\ep$, and not on
$\delta$.
\begin{theorem}
\label{t:interactive}
For $\delta,\ep, \tau \in (0,1)$, we have
\[
C_d(\delta, \ep,\tau) \ge
\frac{d}{\tau^2}\left((1-\delta)\log\frac{1}{\varepsilon} -1 \right).
\]
\end{theorem}
The proof is provided in Section~\ref{ss:interacTestingLB}.

We note that while the lower bound above extends the bounds from the
previous section to the interactive setting, it does not yield optimal
constants for $d=1$ and $r=1$ unlike Theorem~\ref{t:SS1}. In fact, we
believe that even the lower bound in Theorem~\ref{t:ddim} yields a
tight constant; the slackness in characterization of $C_d^1(\delta,
\ep, \tau)$ arises from our upper bound. Thus, the lower bound in
Theorem~\ref{t:interactive} is weaker than those given in the previous
section for $r=1$.

Recall that the lower bounds of the previous section are derived using
the concepts of hypercontractivity and reverse hypercontractivity
(which appeared in the preliminary version of this
paper~\cite{krs-ht-isit18}). As mentioned above, the lower bound in Theorem~\ref{t:interactive} uses a related but different idea of strong data processing
inequalities. In particular, the following bound was derived
in~\cite{hadar2019communication} using a strong data processing inequality; the
statement follows by combining Theorems $7.1$ and $7.2$ in
\cite{hadar2019communication}.
 \begin{lemma}(see~\cite{hadar2019communication})
For $\rho\in [-1,1]^d$ and any interactive communication protocol $\pi$
with inputs $X^n$ and $Y^n$ for parties $\cP_1$ and $\cP_2$,
respectively, we have
\label{l:HadarShayevitz}
\[
D(\mathbb{P}_\rho\|\mathbb{P}_0) \le \rho_{\max}^2\, |\pi|,
\]
where $\rho_{\max}^2 = \max_{i \in [d]} \rho(i)^2$, $\mathbb{P}_\rho$ denotes the distribution in~\eqref{e:model} and $D(P||Q)$ denotes the Kullback-Leibler divergence between distributions $P$ and $Q$.
\end{lemma}
 Using this bound and Fano's inequality, we derive the following lower
 bound for $\widetilde{C}_d(\tau)$.
\begin{theorem} There exists a constant $c>0$, such that for every $\tau \in (0,1)$,
\[
\widetilde{C}_d(\tau) \geq \frac{c\, d^2}{\tau^2}.
\]
\end{theorem}
The proof is provided in Section \ref{ss:estimationLB}.

In fact, the lower bound above is tight too, and matches the upper
bound attained by the distributed estimate proposed
in~\cite{hadar2019distributed}. The lower bound above establishes that
a simple estimate-and-test approach using the estimate
in~\cite{hadar2019distributed} or other estimates will not be able to attain the optimal
$O(d/\tau^2)$ communication needed for correlation testing.
\section{Our scheme and its analysis}\label{s:scheme}
Our general scheme is obtained by first relating the $d$-dimensional
correlation testing problem to the one-dimensional correlation testing
problem, and then relating the one-dimensional problem to its
one-sided version. We develop a test for this one-sided problem first
and then, in steps, convert it to a test for the $d$-dimensional
problem in separate subsections below.

\subsection{One-sided correlation test}
Consider the following one-sided variant of the correlation testing
problem with $d=1$:
\begin{align*}
  \mathcal{H}_0^+ &: \rho \ge \tau, \\ \mathcal{H}_1^1 &: \rho=0,
\end{align*}
where $\tau \in (0,1]$ is known to both parties. We present a $1$-interactive distributed test for this problem; namely, we present a test using one-way communication from $\cP_1$ to $\cP_2$.

Specifically, fix parameters $r>0$, $\theta \leq \tau$, and $k\in
\mathbb{N}$.  Throughout this section, for brevity, with a slight
abuse of notation we denote by $X=(X_1, ..., X_n)\in \mR^n$ and
$Y=(Y_1, ..., Y_n)\in \mR^n$, respectively, the observation of $\cP_1$
and $\cP_2$, where $(X_t, Y_t)_{t=1}^n$ are generated i.i.d. from the
distribution in~\eqref{e:model}. Furthermore, for two vectors $u$ and
$v$ in $\mR^n$, we denote $u\cdot v:= u^T v$.
\begin{enumerate}
\item Using the shared randomness, parties generate an $n \times 2^k$
  matrix $U$ consisting of i.i.d. uniform $\{-1, +1\}$-valued entries
  $U_{ij}$, $1\leq i \leq n$, $1\leq j\leq 2^k$.

\item Let $U_j$ denote the $j$-th column of $U$. $\mathcal{P}_1$ finds
  the \emph{least} index $j\in [2^k]$ such that
\[
U_j \cdot X \ge r\sqrt{n},
  \]
and sends the $k$-bit representation of $j$ to $\mathcal{P}_2$. If no
such $j$ is found, declare $\mathcal{H}_1^1$.

\item $\mathcal{P}_2$, upon receiving $j$, declares $\mathcal{H}_0^+$
  if
\[
U_j\cdot Y \geq \theta\cdot r\sqrt{n}.
\]
\end{enumerate}
The next result captures the performance of our proposed distributed
test.
\begin{theorem}\label{t:one-sided}
  For $\delta, \ep \in (0,1)$, $\tau\in (0,1)$, an
  appropriate choice of $\theta \leq \tau$, and for all $n$
  sufficiently large, the $1$-interactive test proposed above
  satisfies
\begin{align}
  \mathbb{P}_{\mathcal{H}_0^+}\left[{\text{Declare
      }\mathcal{H}_1^1}\right]\leq \delta\,\, \text{ and }\,\,
  \mathbb{P}_{\mathcal{H}_1^1}\left[{\text{Declare
      }\mathcal{H}_0^+}\right]\leq \ep,
  \label{e:error_condition}
\end{align}
when $r$ is set as follows:
\[
r^2 =\frac{2\ln 2}{\tau^2}\left(\sqrt{\log\frac 1 \ep + \log \ln \frac
  3 \delta + 1}+\sqrt{(1-\tau^2) \log \frac 3 \delta}\right)^2,
\]
and the communication length $k$ satisfies
 \begin{align*}
k =\left\lceil \log \frac 1 {Q(r)} + \log \ln \frac 3
\delta\right\rceil,
 \end{align*}
where $Q(\cdot)$ denotes the complementary cumulative distribution
function of a standard Gaussian random variable.
\label{t:1int}
\end{theorem}
\begin{proof}
We begin by deriving a lower bound for the probability of correctly
declaring $\mathcal{H}_0^+$. We have
\begin{align*}
\mathbb{P}_{\mathcal{H}_0^+}\left[{\text{Declare
    }\mathcal{H}_0^+}\right] &=\sum_{j=1}^{2^k}
\mathbb{P}_{\mathcal{H}_0^+} \Big(U_l\cdot X < r\sqrt{n} \text{ for
  all } l \le j-1,U_j\cdot X \ge r\sqrt{n}, U_j\cdot Y \ge \theta\cdot
r\sqrt{n}\Big),
\end{align*}
where $U_0$ is set to be $0$. We approximate the right-side using the Berry-Esseen theorem
($cf.$~\cite{durrett2019probability}) for a fixed realization $X =x$.
Specifically, noting that $U_j\cdot x = \sum_{i=1}^n U_{ij}x_i$ is a
sum of independent random variables, the Berry-Esseen theorem yields
\begin{align*}
\lefteqn{\mathbb{P}_{\mathcal{H}_0^+}\left(U_l\cdot X < r \sqrt{n}
  \text{ for all } l \le j-1, U_j\cdot X \ge r\sqrt{n} \Big|X = x
  \right)} \\ &=
\left[\mathbb{P}_{\mathcal{H}_0^+}\left(\sum_{i=1}^nU_{i1}x_i < r
  \sqrt{n}\Big|X =
  x\right)\right]^{j-1}\mathbb{P}_{\mathcal{H}_0^+}\left(\sum_{i=1}^nU_{ij}x_i
\ge r\sqrt{n}\Big|X = x \right)\\ &\ge
\left(1-Q\left(\frac{r\sqrt{n}}{\sqrt{\sum_{i=1}^n
    x_i^2}}\right)-c_0\frac{\sum_{i=1}^n |x_i|^3}{(\sum_{i=1}^n
  x_i^2)^{\frac
    32}}\right)^{j-1}\left(Q\left(\frac{r\sqrt{n}}{\sqrt{\sum_{i=1}^n
    x_i^2}}\right)-c_0\frac{\sum_{i=1}^n |x_i|^3}{(\sum_{i=1}^n
  x_i^2)^{\frac 32}}\right),
\end{align*}
where $c_0$ is a constant.
Next, note that under $\mathcal{H}_0^+$, for each $i\in [n]$ we have
$\mathbb{E}[Y_i|X_i]=\rho X_i$ with $\rho\geq \tau$. It follows that
for a fixed realization $X = x$ and $U = u$, the random variables
$U_{ij}Y_i$, $1\leq i \leq n$, are independent with distribution
$\cN(\rho x_i, 1-\rho^2)$ for every $j \in [2^k]$.  Note that for
$u_j\cdot x\ge r\sqrt{n}$, we have
\[
\bEE{U_j\cdot Y\mid U=u, X=x}= \rho (u_j\cdot x)\geq \rho r \sqrt{n}.
\]
Therefore, for every $u$ and $x$ such that $u_j \cdot x \ge r\sqrt{n}$
and $u_l \cdot x <r\sqrt{n}$ for all $l\le j-1$, we have
\begin{align*}
\mathbb{P}_{\mathcal{H}_0^+}\left(U_j\cdot Y \ge\theta
r\sqrt{n}\Big|U=u, X=x\right) \ge
Q\left(\frac{r\left(\theta-\rho\right)}{\sqrt{1-\rho^2}}\right) \ge
Q\left(\frac{r\left(\theta-\tau\right)}{\sqrt{1-\tau^2}}\right),
\end{align*}
where the final bound holds since $Q(a)$ is decreasing in $a$ and the
function $f(a)=(\theta-a)/\sqrt{1-a^2}$ is non increasing in $a$ for
$a\geq \theta$; specifically, this bound uses our assumption that
$\theta\leq \tau$.

Upon combining the bounds above, denoting
$\sigma_n(X)=\sqrt{\sum_{i=1}^n X_i^2}$ and $\beta_n(X)=c_0\sum_{i=1}^n
|X_i|^3/\sigma_n^{3}(X)$, we obtain
\begin{align*}
\lefteqn{{\mathbb{P}_{\mathcal{H}_0^+}\left[{\text{Declare
        }\mathcal{H}_0^+}\right]}} \\ &\geq
\bEE{\frac{Q\left(\frac{r\sqrt{n}}{\sigma_n(X)}\right)-\beta_n(X)}{Q\left(\frac{r\sqrt{n}}{\sigma_n(X)}\right)+\beta_n(X)}\times
  \left(1-
  \left(1-Q\left(\frac{r\sqrt{n}}{\sigma_n(X)}\right)-\beta_n(X)\right)^{2^k}\right)}
Q\left(\frac{r\left(\theta-\tau\right)}{\sqrt{1-\tau^2}}\right).
\end{align*}
Using the law of large numbers and the inequality $1-a\leq e^{-a}$,
for every $\eta>0$ and all $n$ sufficiently large, we get
\begin{align}
\mathbb{P}_{\mathcal{H}_0^+}\left[{\text{Declare
    }\mathcal{H}_0^+}\right] &\ge
(1-\eta)\left(1-e^{-2^kQ(r)}\right)Q\left(\frac{r(\theta-\tau)}{\sqrt{1-\tau^2}}\right)
\nonumber \\ &\geq 1 -e^{-2^kQ(r)} -
Q\left(\frac{r(\tau-\theta)}{\sqrt{1-\tau^2}}\right) - \eta,
\label{eq:h0}
\end{align}
where we used the bound $(1-x)(1-y)(1-z)\geq 1-(x+y+z)$ for
$x,y,z\in(0,1)$.

Next, we derive an upper bound for the probability of declaring
$\cH_0^+$ when $\mathcal{H}_1^1$ is true; we derive a bound for this
probability which holds for every fixed realization $u$ of the random
codebook $U$. Since $\sum_{i=1}^nu_{ij}X_i$ is a sum of $n$
independent standard Gaussian random variables, we have
\[
\mathbb{P}_{\mathcal{H}_1^1}\left({\sum_{i=1}^nu_{ij}X_i\ge r \sqrt{n}
  \Big|U=u}\right) \le Q(r),
\]
and similarly,
\[
\mathbb{P}_{\mathcal{H}_1^1}\left({\sum_{i=1}^n u_{ij}Y_i \ge \theta r
  \sqrt{n}\Big|U=u}\right) \le Q(\theta r).
\] 
Therefore,
\begin{align}
\mathbb{P}_{\mathcal{H}_1^1}\left(\text{Declare } \cH_0^+\right) &\le
\mathbb{E}_U\left[\sum_{j=1}^{2^k}
  \mathbb{P}_{\mathcal{H}_1^1}\left(\sum_{i=1}^nu_{ij}X_i\ge r
  \sqrt{n}\right) \cdot
  \mathbb{P}_{\mathcal{H}_1^1}\left(\sum_{i=1}^nu_{ij}Y_i\ge \theta r
  \sqrt{n}\right) \right] \nonumber \\ &\le 2^{k}\, Q(r)\, Q(\theta
r).
\label{eq:h1}
\end{align}
To satisfy the error condition~\eqref{e:error_condition}, by
\eqref{eq:h0} and \eqref{eq:h1} it suffices to set $\eta=\delta/3$ and
choose $r,\theta$, and $k$ to satisfy the following:
\begin{align}
\ln \frac 3 \delta &\leq 2^kQ(r) \leq 2\ln \frac 3 \delta,
\label{e:condition1}
\\ Q\left(\frac{r(\tau-\theta)}{\sqrt{1-\tau^2}}\right) &\leq
\frac{\delta}{3},
\label{e:condition2}
\\ Q(\theta r) &\leq \frac{\ep}{2\ln \frac 3 \delta}.
\label{e:condition3}
\end{align}
Using Chernoff bound $Q(x)\leq e^{-x^2/2}$, for
conditions~\eqref{e:condition2} and~\eqref{e:condition3} it suffices
to have
\begin{align*}
\frac{1-\tau^2}{(\tau-\theta)^2}\cdot \log \frac 3 \delta \leq
\frac{r^2}{2\ln 2}, \\ \frac 1 {\theta^2}\left(\log\frac 1 \ep + \log
\ln \frac 3 \delta + 1\right) \leq \frac{r^2}{2\ln 2}.
\end{align*}
Therefore, the least value of $k$ is given by an $r$ that satisfies
\begin{align*}
\frac{r^2}{2\ln 2} = \min_{\theta \le \tau}
\max\left\{\frac{a}{\theta^2},
\frac{b}{\left(\tau-\theta\right)^2}\right\},
\end{align*}
where $a = \left(\log\frac 1 \ep + \log \ln \frac 3 \delta + 1\right)$
and $b = (1-\tau^2) \log \frac 3 \delta$. The optimal $\theta^*$ for
the problem on the right-side is given by
\[
\theta^* = \frac{\tau\sqrt{a}}{\sqrt{b} + \sqrt{a}},
\]
whereby our optimal choice of $r^2$ is
\begin{align*}
\frac{r^2}{2\ln 2} &=
\frac{1}{\tau^2}\left(\sqrt{a}+\sqrt{b}\right)^2\\ &=\frac{1}{\tau^2}\left(\sqrt{\log\frac
  1 \ep + \log \ln \frac 3 \delta + 1}+\sqrt{(1-\tau^2) \log \frac 3
  \delta}\right)^2.
\end{align*}
Thus, by~\eqref{e:condition1}, we can
satisfy~\eqref{e:error_condition} if we set\footnote[2]{In our
  analysis, we cannot set $k$ higher than this either.} $k=\left\lceil\log
\frac 1 {Q(r)} + \log \ln \frac 3 \delta\right\rceil$ for $r$ given above.
\end{proof}

\subsection{Distributed correlation test for $d=1$}
We now extend the one-sided test above to a test for $d=1$. We present
a general reduction which will allow us to use any $1$-interactive
distributed test for the one-sided problem (not just the one above)
for the (two-sided) correlation testing problem with $d=1$.

\begin{lemma}[Two-sided to one-sided]\label{l:reduction-two-one}
  For $\delta\in(0,1)$, $\ep\in (0,1/2)$, $\tau\in (0,1)$, and
  $\ell\in \mN$, suppose that  
  $T^+=T^+(X^n, Y^n)$ is an
 $1$-interactive $(\ell, \delta, \ep, \tau)$-test
  for the one-sided correlation testing problem. Then, we can find a
$1$-interactive  $(\ell, \delta, 2\ep, \tau)$-test for the correlation testing
  problem with $d=1$.
\end{lemma}
\begin{proof}
 We begin by noting that $T^-(X^n, Y^n)=T^+(X^n, -Y^n)$ is an $(\ell,
 \delta, \ep, \tau)$-test for the following alternative one-sided
 problem:
\begin{align*}
\mathcal{H}_0^- &: \rho \le -\tau,\\ \mathcal{H}_1^1 &: \rho = 0.
\end{align*}
Note that the communication protocol for $T^+$ and $T^-$ is the same;
the corresponding decision mappings $g^+$ and $g^-$ differ. In
particular, $g^-(Y^n,\Pi,V) = g^+(-Y^n, \Pi, V)$, and let $\pi$ be the
common communication protocol for $T^+$ and $T^-$. Consider the
following $1$-interactive distributed test $T=(\pi, g)$ for the
correlation testing problem.
\begin{enumerate}
\item Parties execute the communication protocol $\pi$.
\item Use decision mapping $g(Y^n, \Pi, V)=\min\{g^+(Y^n, \Pi,
  V),g^-(Y^n, \Pi, V)\}$.
\end{enumerate}
For this test, we can verify that
\begin{align*}
 \mathbb{P}_{\cH_1^1}\left[g(Y^n, \Pi, V) = 0\right] 
  &= \mathbb{P}_{\cH_1^1}\left[g^+(Y^n, \Pi, V)= 0
    \text{ or } g^-(Y^n, \Pi, V)=0 \right] \\ 
    &\leq \mathbb{P}_{\cH_1^1}\left[g^+(Y^n,\Pi, V)= 0\right]+ \mathbb{P}_{\cH_1^1}\left[g^-(Y^n, \Pi, V)= 0\right] \\ &\leq 2\ep.
\end{align*}
Furthermore, under $\cH_0^1$,
\begin{align*}
\mathbb{P}_{\cH_0^1}\left[g(Y^n, \Pi, V)=1\right]&= \mathbb{P}_{\cH_0^1}\left[g^+(Y^n, \Pi,
    V)=g^-(Y^n, \Pi, V)=1 \right] \\
  &\leq \max\left\{\mathbb{P}_{\cH_0^+}\left[g^+(Y^n, \Pi,
  V)=g^-(Y^n, \Pi, V)=1 \right],\right.
\\
&\qquad \left.\mathbb{P}_{\cH_0^-}\left[g^+(Y^n, \Pi, V)=g^-(Y^n,
    \Pi, V)=1 \right]\right\} \\ &\leq \delta,
\end{align*}
which shows that $T$ constitutes an $(\ell, \delta, 2\ep, \tau)$-test.
\end{proof}
Lemma~\ref{l:reduction-two-one}, Theorem~\ref{t:one-sided}, and the
well-known bound $Q(x)\geq \frac{x}{\sqrt{2\pi}(x^2+1)}e^{-x^2/2}$
yield Theorem~\ref{t:upper_bound1}.

\subsection{Proof of Theorem~\ref{t:upper_boundd}}
Finally, now that we have a correlation test for $d=1$, we complete
the proof of Theorem~\ref{t:upper_boundd} to obtain a test for general $d$. We begin by making a simple observation akin to the ``median trick'' in randomized algorithms.

\begin{lemma}\label{l:median_trick}
 For $\alpha, \beta, \tau\in (0,1)$ with $\alpha+\beta<1$,
suppose that we have
an $r$-interactive $(\ell, \alpha, \beta, \tau)$-test for the $d$-dimensional correlation testing problem. Then, for every $\delta,\ep\in(0,1)$, we can obtain
an $r$-interactive $(m\ell, \delta, \ep, \tau)$-test for the $d$-dimensional correlation testing problem whenever
\[
  m\geq \frac 2{(1-\beta +\alpha)^2} \,\max\left\{ \ln \frac 1\delta, \ln \frac 1\ep\right\}.
  \]
\end{lemma}  
\begin{proof}
We provide proof only for odd $r$; even $r$ can be handled similarly. 
  Consider an $r$-interactive distributed test $T=(\pi, g)$ that satisfies
  \begin{align*}
    \bPPr{\cH_0^d}{g(Y^n, \Pi, V)=1}&\leq \alpha,
\\
    \bPPr{\cH_1^d}{g(Y^n, \Pi, V)=0}&\leq \beta,
\end{align*}
where $\pi$ is a communication protocol of length $\ell$. To construct the desired test, we repeat the test above $m$ times independently.
Specifically, we  first apply the test above to $m$ independent copies of $(X^n, Y^n, V)$ to obtain transcripts $\Pi_1, ..., \Pi_m$. Note that the resulting communication protocol is still an $r$-round protocol, with length $m\ell$. Denote by $V_1, ..., V_m$ the independent copies of the shared randomness used for the protocol. Further, denote by
$D_i$ the output $g(Y_{n(i-1)+1}^{ni}, \Pi_i, V_i)$, $1\leq i \leq m$, for the $i$-th copy of the test. Consider the new decision mapping $g^m$ given by
\[
g^m(Y^{nm}, \Pi^m, V^m)=\indicator\left\{\sum_{i=1}^m D_i > mt\right\},
\]
for a fixed $\alpha < t < 1-\beta$.
Note that $D_1, ..., D_m$ are independent bits and by our assumption about $T$, 
satisfy
\begin{align*}
  \bPPr{\mathcal{H}_0^d}{D_i=1}&\leq \alpha,
  \\
  \bPPr{\mathcal{H}_1^d}{D_i=1}&\geq 1-\beta,
\end{align*}
for every $1\leq i \leq m$. 
Therefore, by Hoeffding's inequality, 
\begin{align*}
\bPPr{\mathcal{H}_0^d}{g^m(Y^{nm}, \Pi^m, V^m) =1}
  = \bPPr{\mathcal{H}_0^d}{\sum_{i=1}^m D_i > mt}
\le e^{-2m(t-\alpha)^2},
\end{align*}
and similarly, 
\begin{align*}
  \bPPr{\mathcal{H}_1^d}{g^m(Y^{nm}, \Pi^m, V^m)=0}
= \bPPr{\mathcal{H}_1^d}{\sum_{i=1}^m D_i \le mt}
\le e^{-2m(1-\beta-t)^2}.
\end{align*}
In particular, by setting $m\geq \frac 2{(1-\beta +\alpha)^2} \max\left\{\ln \frac 1\delta,\ln \frac 1\ep\right\}$ and $t=(1-\beta+\alpha)/2$, we obtain the desired test. 
\end{proof}  
Thus, it suffices to construct a distributed test with constant probability of error. We do that in the result below by using a $1$-interactive distributed test for $d=1$. Our test uses a randomized construction; to facilitate its analysis, we note the following fact.
\begin{lemma}\label{l:random_R}
 For $R=\frac 1{\sqrt{d}}\, W$ with $W$ a random vector consisting of i.i.d. Rademacher entries, for every vector $x\in \mR^d$ we have, 
  \[
\bPr{\left(R^Tx\right)^2 \geq \frac{\|x\|_2^2}{2d}} \geq \frac {1}{28}.
  \]
\end{lemma}  
\begin{proof}
  The proof uses the Paley-Zygmund inequality. Specifically,
denote by $Z$ the random variable $R^Tx$. Then,
\begin{align*}
\bEE{Z^2}&=\bEE{\left(\sum_{j=1}^d R_jx_j\right)^2}\\
&= \bEE{\frac{1}{d}\|x\|_2^2 + \sum_{i=1}^d \sum_{j=1}^dR_iR_jx_ix_j
  \indicator\{j\neq i\}}
\\
&= \frac{1}{d}\, \|x\|_2^2,
\end{align*}
where the last step follows from the fact that entries of $R$
are independent with zero-mean. Next, we consider $\bEE{Z^4}$. Note that the only terms in the expansion of $\left(\sum_{i=1}^d R_ix_i\right)^4$ that have nonzero mean are those which have only even powers of entries of $R$. In particular, these are terms of the form $R_i^4x_i^4$ and $R_i^2R_j^2X_i^2X_j^2$ with distinct $i,j$. Therefore, we have
\begin{align*}
  \bEE{Z^4}&= \frac 1 {d^2}\sum_{i=1}^d x_i^4+
      {4\choose 2}  \frac 1 {d^2} \sum_{i=1}^d\sum_{j=1}^dx_i^2x_j^2\indicator\{i\neq j\}
      \\
      &\leq \frac 1{d^2}\,(\|x\|_4^4+6\|x\|_2^4)
      \\
      &\leq \frac {7\|x\|_2^4}{d^2},
\end{align*}
where the final inequality uses $\|x\|_4\leq \|x\|_2$. 
Therefore, by the Paley-Zygmund inequality, for $\nu \in (0,1)$,
\[
\mathbb{P}\left(Z^2 > \nu \bEE{Z^2}\right) \ge (1-\nu)^2
\frac{\bEE{Z^2}^2}{\bEE{Z^4}}
\geq \frac{(1-\nu)^2 }{7}.
\]
The claim follows by setting $\nu = 1/2$.
\end{proof}
We are now in a position to complete the proof of Theorem~\ref{t:upper_boundd}.
We use the distributed test for $d=1$ from Theorem~\ref{t:upper_bound1} to build a test for a general $d$. Specifically, we replace the $d$-dimensional observations $X_1, ..., X_n$ of $\cP_1$ with one-dimensional $\widetilde{X}_1, ..., \widetilde{X}_n$ given by $\widetilde{X}_t= R^TX_t$, $1\leq t\leq n$, where $R$ is a random vector generated as in Lemma~\ref{l:random_R}. Note that $(\widetilde{X}_t, Y_t)_{t=1}^n$ are i.i.d. with
\begin{align*}
\bEE{Y_1\widetilde{X}_1\mid R} &= \bEE{\left(\rho^TX_1 + \sqrt{1-\|\rho\|_2^2}Z_1\right)\left(R^TX_1\right)\mid R}\\
&= \rho^T\bEE{X_1X_1^T}R
\\
&= \rho^TR.
\end{align*}
Thus, by Lemma~\ref{l:random_R},
\[
\bP{R}{\left\{r: \left|\bEE{Y_1\widetilde{X}_1|R=r}\right|\geq \frac{\|\rho\|_2}{\sqrt{2d}}\right\}}
\geq \frac 1 {28}.
\]
Denoting $\cG:=\left\{r: \left|\bEE{Y_1\widetilde{X}_1|R=r}\right| \geq \frac{\|\rho\|_2}{\sqrt{2d}}\right\}$ and $\widetilde{\rho}(r):=\left|\bEE{Y_1\widetilde{X}_1 \mid R=r}\right|$, for every $r\in \cG$ we have
\begin{align*}
  \widetilde{\rho}(r) &\geq \tau/\sqrt{2d} \text{ under }\cH_0^d,
\\
  \widetilde{\rho}(r)&=0\text{ under }\cH_1^d.
\end{align*}
Also, in the test we construct for the $d$-dimensional case, we invoke a $1$-interactive $(\ell, 1/56, 1/112,\tau/\sqrt{2d})$-test $T_1$ for the one-dimensional correlation testing problem $\widetilde{\rho}(r)\geq \tau/\sqrt{2d}$ versus $\widetilde{\rho}(r)=0$ with
\[
\ell\leq \frac{cd}{\tau^2},
\]
for an appropriate constant $c$, as guaranteed by Theorem~\ref{t:upper_bound1}.

Next, consider the test for $\cH_0^d$ versus $\cH_1^d$
that samples $R$ from shared randomness executes the aforementioned test $T_1$ for $\cH_0^d$ versus $\cH_1^d$ the one-dimensional problem $\widetilde{\rho}(R)\geq \tau/\sqrt{2d}$ versus $\widetilde{\rho}(R)=0$.
We make the observation that $\widetilde{\rho}(R)=0$ {\em almost
    surely} for $R$, when $\rho=0$. Thus, the missed detection
  probability for the one-dimensional test remains unchanged. However,
  a false alarm may be raised when $R \notin \cG$ or when the one-dimensional test
  raises a false alarm. It follows that for this test
\[
\bPPr{\cH_0^d}{\text{Declare } \cH_1^d}\leq \frac 1 {56}+ \bPr{R\notin \cG}\leq \frac {55}{56},
\]
and
\[
\bPPr{\cH_1^d}{\text{Declare } \cH_0^d}\leq \frac {1}{112},
\]
whereby it constitutes a $({cd}/{\tau^2}, 55/56, 1/112, \tau)$-test for the $d$-dimensional correlation testing problem.

Thus, we have obtained our desired test with constant probability of error guarantees.
Theorem~\ref{t:upper_boundd} follows by using this test along with Lemma~\ref{l:median_trick}.

\section{Proof of lower bounds for $r=1$}~\label{s:proofLB}
We begin by deriving lower bounds for the one-dimensional problem. Our
lower bounds involve the notions of hypercontractivity and reverse hypercontractivity ($cf.$ \cite{AhlswedeGacs76,
  nair2014equivalent, Borell82, Mossel13}), which we define first.

For $1\leq q\leq p <\infty$, a pair of random
variables $(X,Y)$ is $(p,q)$-hypercontractive if for all $\mR$-valued
functions $f$ of $X$ and $g$ of $Y$,
\[
\bEE{|f(X)g(Y)|} \leq \|f(X)\|_{p^\prime}\|g(Y)\|_q,
\]
 where $p^\prime = p/(p-1)$ is the H\"older conjugate of
 $p$. Similarly, for $1 \ge q > p$, a pair of random variables $(X,Y)$
 is $(p,q)$- reverse hypercontractive if for all $\mR$-valued
 functions $f$ of $X$ and $g$ of $Y$,
\[
\bEE{|f(X)g(Y)|} \geq \|f(X)\|_{p^\prime}\|g(Y)\|_q.
\]
The set of all $(p,q)$ for which $(X,Y)$ is $(p,q)$-hypercontractive
and $(p,q)$-reverse hypercontractive, respectively, are called the
hypercontractivity ribbon and the reverse hypercontractivity ribbon of
$(X,Y)$. The following \emph{tensorization}
property of the hypercontractivity and the reverse hypercontractivity
ribbon is well known.  
\begin{lemma}[Tensorization~\cite{nair2014equivalent}~\cite{Mossel13}]
\label{l:tensor}
For $p \ge 1$, define
\[
q_p(X,Y) = \inf\{q:(X,Y) \text{ is } (p,q)\text{-hypercontractive} \},
\]
and $r_p(X,Y) = q_p(X,Y)/p$. If $(X_i,Y_i)_{i=1}^n$ are independent, then
\begin{align}
  r_p(X^n,Y^n) = \max_{1\leq i\leq n} r_p(X_i,Y_i).
  \nonumber
\end{align}
Furthermore, for $p\le 1$, define
\[
q_p(X,Y) = \sup\{q:(X,Y) \text{ is } (p,q)\text{-reverse hypercontractive}\},
\]
and $s_p(X,Y) = q_p(X,Y)/p$. If $(X_i,Y_i)_{i=1}^n$ are independent, then
\begin{align}
  s_p(X^n,Y^n) = \max_{1\leq i \leq n} s_p(X_i,Y_i).
  \nonumber
\end{align}
\end{lemma}  

We use the notions of hypercontractivity and reverse
hypercontractivity to obtain the change of measure bounds between the
joint distribution and the independent distribution, which in turn
lead to the following two lower bounds for $C_1^1(\delta, \ep, \tau)$. 
\begin{theorem}[Lower bound $1$]\label{t:lower_bound1}
Given $\delta, \ep \in (0,1)$ and $(p, q)$ such that $1\leq p^\prime
\leq q \leq p$ and $(X,Y)$ is $(p,q)$-hypercontractive, the minimum
one-way communication 
 for one-dimensional correlation testing $C_1^1(\delta, \ep, \tau)$ is bounded
below as
\begin{equation}
C_1^1(\delta, \ep, \tau) \geq \frac p q \log \frac 1 \ep - p \log \frac 1
{1-\delta}.
\label{eqn:lb}
\end{equation}
\end{theorem}
\begin{proof}
For $1\leq q \leq p$, suppose that $(X,Y)$ is
$(p,q)$-hypercontractive, which by Lemma~\ref{l:tensor} implies that
$(X^n, Y^n)$ is $(p,q)$-hypercontractive. Furthermore, assume that $p^\prime \leq q$
which is the same as $q^\prime \leq p$. Then, for any subset
$\cA\subset \cX^n$ and $\cB \subset \cY^n$, we have
\begin{align}
\bP{X^nY^n}{\cA \times \cB} \leq \bP{X^n}{\cA}^{\frac 1{p^\prime}}
\bP{Y^n}{\cB}^{\frac 1 q}.
\label{e:hyper_sets}
\end{align}
We begin by considering a deterministic test where the shared randomness $U$
is constant. Specifically, given a deterministic $(\ell, \delta,
\ep,\tau)$-test $T=(f,g)$, denoting\footnote{With a slight abuse of notation, we denote the one-way communication protocol by a mapping $f$.} $L=2^\ell$, let $\cA_i = f^{-1}(i)$ for $i
=1, ..., L$. Then, $\{\cA_1, ..., \cA_L\}$ constitutes a partition of
$\cX^n$. Further, let $\cB_i$ denote the set $\{\by\in \cY^n: g(\by,
i) = 0\}$, namely the set of $\by$ where $\cP_2$ declares $\cH_0^1$ upon
receiving $i$ from $\cP_1$. It follows that
\begin{align*}
1-\delta & \le \sum_{i=1}^L \bP{X^nY^n}{\cA_i \times \cB_i} \\ & \le
\sum_{i=1}^L \bP{X^n}{\cA_i}^{\frac 1
  {p^\prime}}\bP{Y^n}{\cB_i}^{\frac 1q},
\end{align*}
where the previous inequality uses \eqref{e:hyper_sets}. Upon bounding
the right-side using H\"older's inequality, we get
\begin{align*}
1-\delta &\le \sum_{i=1}^L
\left(\bP{X^n}{\cA_i}\bP{Y^n}{\cB_i}\right)^{\frac 1 q}
\bP{X^n}{\cA_i}^{\frac{1}{p^\prime} - \frac{1}{q}} \\ &\le
\left(\sum_{i=1}^L \bP{X^n}{\cA_i}\bP{Y^n}{\cB_i}\right)^{\frac 1q}
\left(\sum_{i=1}^L \bP{X^n}{\cA_i}^{q^\prime\left(\frac{1}{p^\prime} -
  \frac{1}{q}\right)}\right)^{\frac 1 {q^\prime}} \\ &\le \ep^{\frac 1
  q}\left(\sum_{i=1}^L \bP{X^n}{\cA_i}^{q^\prime\left(\frac{1}{p^\prime} -
  \frac{1}{q}\right)}\right)^{\frac 1 {q^\prime}},
\end{align*}
where the previous inequality uses the requirement
$\bPPr{\cH_1^1}{\text{Declare }\cH_0^1}\leq \ep$. Noting that
$q^\prime(1/p^\prime - 1/q) = 1- q^\prime/p$, the assumption $q^\prime
\leq p$ and H\"older's inequality imply
\[
\sum_{i=1}^L \bP{X^n}{\cA_i}^{q^\prime\left(\frac{1}{p^\prime} -
  \frac{1}{q}\right)}\leq L^{\frac {q^\prime}p}.
\]
Combining the bounds above, we get
\[
1-\delta \leq \ep^{\frac 1 q}L^{\frac 1p},
\]
which completes the proof.

When shared randomness is available, we follow the procedure above for
the deterministic test obtained by conditioning on the shared
randomness $V$; let $(\cA^V_i,\cB^V_i)$, $1\leq i \leq L$, denote the
counterpart of $(\cA_i, \cB_i)$ above for shared randomness
$V$. Proceeding as before, we have
\begin{align*}
1-\delta &\le \mathbb{E}
\left[\left(\sum_{i=1}^L\bP{X^n}{\cA^V_i}\bP{Y^n}{\cB^V_i}\right)^{\frac
    1q} \left(\sum_{i=1}^L
  \bP{X^n}{\cA^V_i}^{q^\prime\left(\frac{1}{p^\prime} -
    \frac{1}{q}\right)}\right)^{\frac 1 {q^\prime}}\right] \\ &\leq
\bEE{\left(\sum_{i=1}^L\bP{X^n}{\cA^V_i}\bP{Y^n}{\cB^V_i}\right)^{\frac
    1q}}L^{\frac 1p}.
\end{align*}
It follows from Jensen's inequality that
\begin{align*}
1-\delta &\leq
\bEE{\sum_{i=1}^L\bP{X^n}{\cA^V_i}\bP{Y^n}{\cB^V_i}}^{\frac
  1q}L^{\frac 1p}.  \\ &\leq \ep^{\frac 1q}L^{\frac 1p},
\end{align*}
which completes the proof of Theorem~\ref{t:lower_bound1}.
\end{proof}

\begin{theorem}[Lower bound $2$]\label{t:lower_bound2}
Given $\delta, \ep \in (0,1)$ and $(p, q)$ such that $1\ge q \ge 0 \ge
q^\prime \ge p$ and $(X,Y)$ is $(p,q)$-reverse hypercontractive, the
minimum one-way communication for one-dimensional correlation testing $C_1^1(\delta, \ep, \tau)$ is
bounded below as
\begin{equation}
C_1^1(\delta, \ep, \tau) \geq \frac p q \log \frac 1 {1-\ep} - p \log \frac 1
{\delta}.
\label{eqn:alb}
\end{equation}
\end{theorem}
\begin{proof}
For $1\geq q > p$, suppose that $(X,Y)$ is $(p,q)$-reverse
hypercontractive, which with Lemma~\ref{l:tensor} implies that
$(X^n, Y^n)$ is $(p,q)$-reverse hypercontractive.
Furthermore, assume that $q^\prime \geq p$. Then,
for any subset $\cA\subset \cX^n$ and $\cB \subset \cY^n$, for $0 \le
\theta \le 1$ we have
\begin{align}
\bP{X^nY^n}{\cA \times \cB}^\theta \geq
\bP{X^n}{\cA}^{\theta\left(\frac {p-1}{p}\right)}
\bP{Y^n}{\cB}^{\theta \frac 1 q}.
\label{e:hyper_sets2}
\end{align}
We only provide a proof for deterministic tests; the extension to the
case when shared randomness is used can be completed as in the
proof of Theorem~\ref{t:lower_bound1}. Given a deterministic $(\ell,
\delta, \ep,\tau)$-test $T=(f,g)$, let $\cA_i = f^{-1}(i)$ for $i =1, ..., L=2^\ell$, and let $\cB_i$ denote the set $\{\by\in \cY^n: g(\by, i) = 0\}$. It follows that
\begin{align*}
1-\ep & \le \sum_{i=1}^L \bP{X^n}{\cA_i}\bP{Y^n}{\cB_i}\\ & \le
\sum_{i=1}^L \bP{X^n}{\cA_i}^{1-\theta\left(\frac {p-1}{p}\right)}
\bP{Y^n}{\cB_i}^{1-\frac \theta q} \bP{X^nY^n}{\cA_i \times
  \cB_i}^\theta,
\end{align*}
where the previous inequality uses \eqref{e:hyper_sets2}. Upon
bounding the right-side using H\"older's inequality, we get
\begin{align*}
1-\ep &\le \left(\sum_{i=1}^L\left(
\bP{X^n}{\cA_i}^{1-\theta\left(\frac{p-1}{p}\right)}\bP{Y^n}{\cB_i}^{1-\frac
  \theta q}\right)^{\frac{1}{1-\theta}}\right)^{1-\theta}
\left(\sum_{i=1}^L \bP{X^nY^n}{\cA_i \times \cB_i}
\right)^\theta\\ &\le \left(\sum_{i=1}^L
\bP{X^n}{\cA_i}^{1+\left(\frac
  {\theta}{p(1-\theta)}\right)}\bP{Y^n}{\cB_i}^{\frac{q-\theta}{q(1-\theta)}}\right)^{1-\theta}
\delta^{\theta},
\end{align*}
where the previous inequality uses the requirement
$\bPPr{\cH_0^1}{\text{Declare }\cH_1^1}\leq \delta$. Choosing $\theta=q$,
the assumption $q^\prime \ge p$ together with H\"older's inequality
implies
\[
\left(\sum_{i=1}^L
\bP{X^n}{\cA_i}^{1+\frac{q}{p(1-q)}}\right)^{1-q}\leq L^{\frac
  {-q}{p}}.
\]
Combining the bounds above, we get
\[
1-\ep \leq \delta^{q}L^{\frac{-q}p},
\]
which completes the proof.
\end{proof}  

To obtain tight lower bounds for one-dimensional $X$ and $Y$ jointly
Gaussian, we need to optimize our lower bounds over the entire 
hypercontractivity and reverse hypercontractivity ribbon. We rely on
the following characterizations of the hypercontractivity and the
reverse hypercontractivity ribbons. 
\begin{lemma}[$cf.$~\cite{Gross75}]\label{l:gross}
Let $X$ and $Y$ be one-dimensional with joint distribution given by~\eqref{e:model}. For $1\leq q\leq p$, $(X,Y)$ is $(p,
q)$-hypercontractive if 
and only if    
\begin{align}
\frac{q-1}{p-1}\geq \rho^2.
\label{e:BGB}
\end{align}
Furthermore, for $1 \ge q \geq p$, $(X,Y)$ is $(p, q)$-reverse
hypercontractive if and only if  
\begin{align}
\frac{1-q}{1-p} \geq \rho^2.
\label{e:rev}
\end{align}
\end{lemma}
The next corollary is obtained by maximizing the right-sides of
\eqref{eqn:lb} and \eqref{eqn:alb}, respectively, over the set of
$(p,q)$ satisfying \eqref{e:BGB} and \eqref{e:rev}; the upper bound is
from Theorem~\ref{t:upper_bound1}. 

\begin{corollary}\label{c:SS1} For $0 < \tau \le 1$,
\begin{enumerate}
\item for $\delta\leq 1$ and $\ep$ such that $\delta
+ \ep^{\frac{1-\tau}{1+\tau}} \leq 1$,
\[
C_1^1(\delta, \ep, \tau) =
\frac 1{\tau^2} \log \frac 1 \ep+ 
\Theta_\delta\left(\sqrt{
\log \frac 1 \ep}\right);
\]
\item for $\delta,\ep\in(0,1)$,
\[
C_1^1(\delta, \ep, \tau) = \frac {1-\tau^2}{\tau^2} \log \frac 1
\delta+
\Theta_\ep\left(\sqrt{\log \frac 1 \delta}\right),
\]
where the notation $\Theta_x$ denotes that the constant implied by $\Theta$ depends on $x$. 
\end{enumerate}
\end{corollary}
\begin{proof}
Assume first that $\ep^{\frac{1-\tau}{1+\tau}}\leq
1-\delta$. Using the characterization in Lemma~\ref{l:gross}, $(X,Y)$
is $(p,q)$-hypercontractive for any $p$ and $q$ satisfying
\begin{align*}
p &= 1+w,\\ q &= 1+\tau^2w,
\end{align*}
for any $w\geq 0$. Inserting this choice of $(p,q)$ in the lower bound
of Theorem~\ref{t:lower_bound1}, we get for any $(\ell,\delta, \ep, \tau)$-test
that
\[
\ell \ge \frac{1+w}{1+\tau^2 w} \log \frac{1}{\varepsilon} - (1+w)\log
\frac{1}{1-\delta}.
\]
For brevity, we denote $\xi:=\log (1-\delta)/\log \ep$; our
assumption $\ep^{\frac{1-\tau}{1+\tau}} \le 1-\delta$ is
equivalent to $\xi \leq (1-\tau)/(1+\tau)$. To obtain the
tightest lower bound, we maximize $(1+w)/(1+\tau^2w) - \xi (1+w)$
over $w\geq 0$. The maximum is obtained at $w^*$ given by
\[
w^* = \frac 1 {\tau^2}\left(\sqrt{\frac {1-\tau^2}\xi}-1\right),
\]
provided $\xi \leq 1-\tau^2$, which holds since $1-\tau^2 \geq
(1-\tau)/(1+\tau)$. Furthermore, the corresponding optimal $p^*$
and $q^*$ satisfy $p^{*\prime}\leq q^*$ if and only if
\[
\tau^2\leq \left(\sqrt{\frac {1-\tau^2}\xi}-1\right)^2,
\]
which is satisfied when $\xi \leq (1-\tau)/(1+\tau)$. Thus,
\begin{align}
\ell &\ge\log \frac{1}{\varepsilon} \left( \frac{1+w^*}{1+\tau^2 w^*} -
\xi(1+w^*)\right) \nonumber \\ &= \frac 1{\tau^2} \left(\sqrt{\log
  \frac 1\ep} - \sqrt{\big(1- \tau^2\big)\log \frac 1 {1-\delta}}
\right)^2.
\label{e:lower_bound_SS1}
\end{align}
The first part of Corollary~\ref{c:SS1} follows
from~\eqref{e:lower_bound_SS1} and Theorem~\ref{t:upper_bound1}.

To get the second part of Corollary~\ref{c:SS1}, we obtain a
replacement for \eqref{e:lower_bound_SS1} using the reverse
hypercontractivity part of Lemma~\ref{l:gross}.  Specifically, $(X,Y)$
is $(p,q)$-reverse hypercontractive for any $p$ and $q$ satisfying
\begin{align*}
p &= 1-w,\\ q &= 1-\tau^2w,
\end{align*}
for any $\frac{1}{\tau^2} \geq w \geq 0$ since $q$ must be greater
than or equal to $0$. Inserting this choice of $(p,q)$ in the lower
bound of Theorem~\ref{t:lower_bound2}, we get for any $(\ell,\delta,
\ep,\tau)$-test that
\[
\ell \ge \frac{1-w}{1-\tau^2 w} \log \frac{1}{1-\varepsilon} - (1-w)\log
\frac{1}{\delta}.
\]
We maximize the right-side of the above inequality subject to $w \le
\frac{1}{\tau^2}$. The maximum is obtained at $w^*$ given by
\[
w^* = \frac 1 {\tau^2}\left(1-\sqrt{\frac {(1-\tau^2)\log
    \frac{1}{1-\ep}}{\log \frac{1}{\delta}}}\right).
\]
Note that $w^* \le \frac{1}{\tau^2}$ is satisfied for every $\delta$
and $\ep$, and the additional assumption
$\ep^{\frac{1-\tau}{1+\tau}}\leq 1-\delta$ of the first part of
Corollary~\ref{c:SS1} is not required for the second part. Thus,
\begin{align}
\ell &\ge \frac{1-w^*}{1-\tau^2 w^*} \log \frac{1}{1-\ep} - (1-w^*)\log
\frac{1}{\delta} \nonumber \\ &= \frac 1{\tau^2} \left(\sqrt{\log
  \frac 1{1-\ep}} - \sqrt{\big(1- \tau^2\big)\log \frac 1 {\delta}}
\right)^2, \nonumber
\end{align}
which together with Theorem~\ref{t:upper_bound1} yields the second
part of Corollary~\ref{c:SS1}.
\end{proof}

Finally, we exploit tensorization property in Lemma~\ref{l:tensor} to provide a matching lower bound for
Theorem~\ref{t:upper_boundd} in the result below. 
\begin{theorem}
\label{t:ddim}
\begin{enumerate} For $0 < \tau \le 1$,
\item for $\delta\in(0,1)$ with $\ep$ such that
$\delta+ \ep^{\frac{1-\tau}{1+\tau}} \leq 1$, we have 
\[
C_d^1(\delta, \ep, \tau) \ge \frac d{\tau^2} \left(\sqrt{\log \frac 1\ep} - \sqrt{\left(1- \frac{\tau^2}d\right)\log \frac 1 {1-\delta}}\right)^2;
\]
\item for $\delta,\epsilon \in (0,1)$,
\[
C_d^1(\delta, \ep, \tau) \ge \frac d{\tau^2} \left(\sqrt{\log  \frac 1{1-\ep}} - \sqrt{\left(1- \frac {\tau^2}{d}\right)\log \frac 1 {\delta}}\right)^2.
\]
\end{enumerate}
\end{theorem}
\begin{proof}
  We consider a different problem where the observation of $\cP_1$ remains the same but we provide more information to $\cP_2$. Specifically, $\cP_1$ observes i.i.d. copies of $X = (X(1),\ldots,X(d))$ and  $\cP_2$ observes i.i.d. copies of $Y = (Y(1),\ldots,Y(d))$ where for $i=1,\ldots,d$,
\[
\mathbb{E}[Y(i)|X] = \rho(i) X(i).
\]
Note that in our original problem the observation of $\cP_2$ are
i.i.d. copies of $Y(1) + \ldots + Y(d)$. 
With this modified observation for $\cP_2$, we consider the hypothesis
testing problem of $\mathcal{H}_0^d$ versus $\mathcal{H}_1^d$ as
before. Denote by $\widetilde{C}_d^1(\delta,\varepsilon,\tau)$ the
the minimum $\ell$ such that we can find a $1$-interactive $(\ell, \delta,
\varepsilon, \tau)$-test for this modified problem. Since the observation for
the former problem can be obtained from the latter problem as well, we have
\[
  {C}_d^1(\delta,\varepsilon,\tau) \geq
  \widetilde{C}_d^1(\delta,\varepsilon,\tau).
  \]
Furthermore, with $\mathcal{X} =
\mathcal{Y} = \mathbb{R}^d$, the proof of
Theorem~\ref{t:lower_bound1} applies to the modified problem as well,
and we obtain the following bound:
\[
\widetilde{C}_d^1(\delta,\varepsilon,\tau) \ge \frac p q \log \frac 1 \ep - p \log \frac 1 {1-\delta},
\]
where $(X^n, Y^n)$ is $(p,q)$-hypercontractive.
By Lemma~\ref{l:tensor} and Lemma~\ref{l:gross}, we can parameterize
$p$ and $q$ as
\begin{align*}
p &= 1+w,\\ q &= 1+\rho_{\max}^2w, 
\end{align*}
with $w \ge 0$ and $\rho_{\max}^2 := \max_{i=1}^d \rho(i)^2$. Proceeding as in the proof of Corollary~\ref{c:SS1}, we get
\[
\widetilde{C}_d^1(\delta,\varepsilon,\tau) \ge \frac 1{\rho_{\max}^2} \left(\sqrt{\log \frac 1\ep} - \sqrt{\big(1- \rho_{\max}^2\big)\log \frac 1 {1-\delta}}\right)^2.
\]
Note that we can choose any $\rho$ such that $\|\rho\|_2\geq \tau$. Among all such $\rho$s, the minimum value of $\rho_{\max}$ is attained by $\rho$ with $\rho(i)^2 =  \tau^2/d$. Using this value for $\rho$, we get
\[
\widetilde{C}_d^1(\delta,\varepsilon,\tau) \ge \frac d{\tau^2} \left(\sqrt{\log \frac 1\ep} - \sqrt{\left(1- \frac{\tau^2}d\right)\log \frac 1 {1-\delta}}\right)^2,
\]
which completes the proof of the first part of Theorem~\ref{t:ddim}. The proof of the second part is completed similarly by using the tensorization property of the reverse hypercontractivity ribbon.
\end{proof}

\section{Lower bounds for interactive communication}\label{s:proofLB2}
\subsection{Proof of Theorem~\ref{t:interactive}}
\label{ss:interacTestingLB}
Let $T=(\pi, g)$ constitute an $(\ell,\delta, \ep, \tau)$-test. Denote
by $P$ the distribution of $(Y^n,\Pi,V)$ under $\mathcal{H}_0^d$ and
by $Q$ the distribution of $(Y^n, \Pi,V)$ under $\mathcal{H}_1^d$.

Then, 
\begin{align*}
  P\left(g(Y^n, \Pi, V)=1\right)
  &\leq \delta ,
  \\
  Q\left(g(Y^n, \Pi, V)=0\right) &\leq \varepsilon.
\end{align*}
Denoting $P\left(g(Y^n, \Pi, V)=0\right)$ and 
$Q\left(g(Y^n, \Pi, V)=0\right)$ by $p_0$ and $p_1$, respectively,
let $P_i$ denote the Bernoulli distribution with parameter $p_i$, $i=0,1$. Then, by the data processing inequality applied using the channel $\indicator\{g(Y^n, \Pi, V)=0\}$, we have
\begin{align}
  D(P\|Q) &\ge D(P_0 || P_1)
  \nonumber
  \\
&\geq 
p_0\log\frac 1 {p_1} -1
\nonumber
\\
&\ge (1-\delta)\log\frac{1}{\varepsilon} -1,
\label{eq:dataprocessing}
\end{align}
where we used the bound $h(p_0)\leq 1$ and $h(\cdot)$ denotes the binary entropy function. Furthermore, by Lemma \ref{l:HadarShayevitz} we have $D(P\|Q) \le \rho_{\max}^2\ell$, which with the previous bound gives
\[
\ell \ge \frac{1}{\rho_{\max}^2}\left((1-\delta)\log\frac{1}{\ep} - 1\right).
\]
By choosing $\rho$ such that $\|\rho\|_2\geq \tau$ and $\rho_{\max}$ is maximized, namely by choosing $\rho(i)^2 =  \tau^2/d$ for every $1\leq i \leq d$, we get the desired bound
\[
\ell \ge \frac{d}{\tau^2}\left[(1-\delta)\log\frac{1}{\ep} - 1\right].
\]

\subsection{Lower Bound for Estimation}
\label{ss:estimationLB}
We provide lower bounds for the estimation error using Fano's method. 
Using the \emph{Gilbert-Varshamov construction} (see, for instance,
\cite[Problem 5.5]{CsiKor11}) we can find $m\geq 2^{d(1-h(1/4))}\geq 2^{d/6}$ vectors $u_1 , ..., u_m \in\{-1,+1\}^d$
such that\footnote{We denote by $d_H(u,v)$ the Hamming distance between $u$ and $v$.} $d_H(u_i, u_j)\geq d/8$ for every $i\neq j$. Therefore, the vectors $\rho_i:= \frac{\Delta}{\sqrt{d}}\cdot u_i$, $1\leq i \leq m$ satisfy 
\begin{align*}
  \min_{i,j \in [m]:i\neq j} \|\rho_i - \rho_j\|_2^2 &=\frac{4\Delta^2d_H(u_i, v_i)}{d}
\geq \frac{\Delta^2}2,
\\
\max_i\rho_j(i)^2 &= \frac{\Delta^2}{d}, \quad \forall\, 1\le j\leq m.
\end{align*}
Consider an $r$-interactive $(\ell, \tau)$-estimate $(\pi, \widehat{\rho})$. 
We use the estimator $\widehat{\rho}$ to resolve between the hypotheses $\mathcal{H}_j, j \in [m]$ where under $\mathcal{H}_j$, $X \in \mathbb{R}^d$ and $Y \in \mathbb{R}$ are jointly Gaussian and
\[
\mathbb{E}[Y(i)|X] = \rho_j(i)X(i), i \in [d].
\]
Consider the test which declares\footnote{In the remainder of this proof, with an abuse of notation, we denote the random variable $\widehat{\rho}(Y^n, \Pi, V)$ by $\widehat{\rho}$.} $\mathcal{H}_j$ if $\|\widehat{\rho}-\rho_j\|_2^2 < \Delta^2/8$; the output is unique since $\|\rho_i-\rho_j\|_2^2\geq\Delta^2/2$ for every $i\neq j$. 
The probability of error for this test under $\mathbb{P}_{\rho_j}$ is bounded above by 
\begin{align*}
  \mathbb{P}_{\rho_j} \left(\|\widehat{\rho}-\rho_j\|_2^2 \ge \Delta^2/8\right)\le \frac{8}{\Delta^2} \mathbb{E}_{\rho_j}\left[\|\widehat{\rho}-\rho_j\|_2^2\right],
\end{align*}
where the inequality is by Markov's inequality. Therefore, denoting by $P_e^*$ the minimum average probability of error for this hypothesis testing problem under uniform prior on the hypotheses, we get from \eqref{eqn:estaccuracy} that
\begin{align*}
\tau^2 &\ge \max_{j \in [m]} \frac{\Delta^2}8 \mathbb{P}_{\rho_j}\left(\|\widehat{\rho}-\rho_j\|_2^2 \ge \Delta^2/8\right)\\
&\ge \frac{\Delta^2}{8m} \sum_{i=1}^{m} \mathbb{P}_{\rho_j}\left(\|\widehat{\rho}-\rho_j\|_2^2 \ge \Delta^2/8\right)\\
&\geq \frac{\Delta^2}{8} P_e^*.
\end{align*}
By Fano's inequality, we have 
\begin{equation}
\label{e:fano}
P_e^* \ge 1 - \frac{C(W)+1}{\log m},
\end{equation}
where $W$ denotes the channel with input $j\in\{1,...,m\}$ and output
$(Y^n,\Pi,V)$ with distribution corresponding to the correlation $\rho_j$ between $X^n$ and $Y^n$, and $C(W)$ denotes the capacity of channel $W$. Recall the well-known bound
\[
C(W) \le \min_Q\max_j D(W(\cdot|j) \| Q).
\]
We use this bound with $Q$ chosen to be the distribution of $(Y^n, \Pi,V)$
when the correlation between $X$ and $Y$ is
$\rho=0$.

  Then, by Lemma \ref{l:HadarShayevitz} we have
\[
D(W(\cdot|j) \| Q) \le \max_{1\leq i \leq d}\rho_j(i)^2 \ell= \frac{\Delta^2\ell}{d}.
\]
Combining the bounds above yields
\[
\tau^2\geq \frac{\Delta^2}8 \left(1- \frac{6(\Delta^2\ell/d + 1)}{d}\right).
\]
In particular, for $d\geq 12$, setting $\Delta^2\ell/d^2=1/24$ gives
$\ell \geq \frac{d^2}{768\tau^2}$. Note that for $d <12$, for an appropriate constant $c$, the lower bound $\ell\geq cd^2\,/\tau^2$ holds
since we already have an $\Omega(d/\tau^2)$ lower bound for the testing problem. This completes the proof.
\qed

\section{Extensions and discussion}
We conclude with a discussion on various extensions of our result, and state some of these extensions without proof (the proofs are very similar to the others in the paper).

First, we note that while the hypercontractivity based lower bound
yields a tight dependence on $\delta$ or $\ep$ separately in
Corollary~\ref{c:SS1}, 
it does not characterize the joint dependence on $\delta$ and $\ep$
simultaneously. Interestingly, when we allow  $\delta>1/2$ and have
$\ep$ small, such a joint characterization is possible. Specifically,
for $d=1$, consider the simple binary hypothesis problem of correlation $\rho$
versus correlation $0$. The test we use in Theorem~\ref{t:one-sided} 
for resolving between $\mathcal{H}_0^+$ and $\mathcal{H}_1^1$ with a
different choice for $\theta$ and $r$ yields a joint
characterization of one-way minimum communication needed for this problem (see~\cite{krs-ht-isit18}).
Interestingly, the overall communication is below
$(1/\rho^2)\max\{\log 1/\ep, \log 1/\delta\}$.
\begin{theorem}
\label{c:SS2}
For $d=1$, $0 < \rho \leq 1$, $\delta \in (1/2,1)$, and $\ep$ such that $\delta + \ep^{\frac{1-\tau}{1+\tau}} \leq 1$,
the minimum one-way communication needed to test if correlation is $\rho$ or $0$ is given by 
 \[
\frac 1{\rho^2} \cdot\left(\sqrt{\log \frac 1\ep}
- \sqrt{\big(1- \rho^2\big)\log \frac 1 {1-\delta}} \right)^2 +
O\left(\sqrt{\log \frac 1{\ep} \log \frac 1{1-\delta}}\right).
\]
\end{theorem}

In another direction, we can consider the simple binary hypothesis
testing problem of $\rho=\rho_0$ versus $\rho=\rho_1$, where
$1>\rho_0> \rho_1>0$. Once again, by modifying the parameters for the
test used in  Theorem~\ref{t:one-sided}, we get a generalization of our
results for $d=1$ to the case $\rho_1>0$.
Specifically, in this case, the probability of error requirements as
in \eqref{eq:h0} and \eqref{eq:h1} yield
\begin{align*}
1-e^{2^kQ(r)} - Q\left(\frac{r(\rho_0-\theta)}{\sqrt{1-\rho_0^2}}\right) - \eta &\ge 1-\delta,\\
2^{k+1}Q(r)Q\left(\frac{r(\theta-\rho_1)}{\sqrt{1-\rho_1^2}}\right) &\le \epsilon.
\end{align*}
Proceeding in a similar manner as our earlier analysis and upon
setting $\theta\in (\rho_1, \rho_0)$ and
\[
r^2 =\frac{2\ln 2}{(\rho_0-\rho_1)^2}\left(\sqrt{(1-\rho_1^2)\log\frac 1 \ep + \log \ln \frac
  3 \delta + 1}+\sqrt{(1-\rho_0^2) \log \frac 3 \delta}\right)^2,
\]
we obtain the following result.
\begin{theorem}
For $d=1$, $\delta, \ep \in (0,1)$, $0 < \rho_1 < \rho_0 < 1$, we can find a
distributed test for $\rho=\rho_0$ versus $\rho=\rho_1$ that uses
one-way communication of less than
\[
\frac{1}{(\rho_0-\rho_1)^2}\left(\sqrt{(1-\rho_1^2)\log\frac 1
  \ep}+\sqrt{(1-\rho_0^2) \log \frac 1 \delta}\right)^2
+O\left(
\sqrt{\log \frac 1 \ep + \log \ln \frac
  1 \delta}\sqrt{\log \frac 1 \delta}
\right)
 \text{ bits}.
\]
\end{theorem}
We note that~\cite{hadar2019error} derived an upper bound for the
error-exponent for this problem when communication length per sample is fixed.
While the result
there  was stated for error exponent, the main
bound~\cite[Equation (48)]{hadar2019error} shows that
the one-way communication needed for testing $\rho=\rho_0$
versus $\rho=\rho_1$ must exceed
\[
\left(\frac{(1-\rho_1)^2}{(\rho_0-\rho_1)^2} -1
\right)\left( \max\left\{(1-\delta)\log \frac 1 \ep,
(1-\ep)\log \frac 1 \delta\right\}-1\right),
\]
which almost matches the communication requirement for our scheme.
However, we do not
account for the number of
samples in our scheme, and it may not attain the upper bound on the error
exponent in~\cite{hadar2019error}. 

We close with the remark that extending our results to the case when
the basis used to express $\bEE{Y_1|X_1=x_1}$ as a linear function of
$x_1$ is unknown is practically relevant and an interesting open
problem. It will account for testing for correlation in some
unknown ``feature representation'', and not along fixed features as in our
current setting.  

\section*{Acknowledgements}
The authors thank Ofer Shayevitz and Uri Hadar for useful discussions and Prathamesh Mayekar for suggesting the use of median trick in the proof of upper bound for general $d$. 
\bibliography{IEEEabrv,references} \bibliographystyle{IEEEtran}

\begin{thebibliography}{10}
\providecommand{\url}[1]{#1}
\csname url@samestyle\endcsname
\providecommand{\newblock}{\relax}
\providecommand{\bibinfo}[2]{#2}
\providecommand{\BIBentrySTDinterwordspacing}{\spaceskip=0pt\relax}
\providecommand{\BIBentryALTinterwordstretchfactor}{4}
\providecommand{\BIBentryALTinterwordspacing}{\spaceskip=\fontdimen2\font plus
\BIBentryALTinterwordstretchfactor\fontdimen3\font minus
  \fontdimen4\font\relax}
\providecommand{\BIBforeignlanguage}[2]{{%
\expandafter\ifx\csname l@#1\endcsname\relax
\typeout{** WARNING: IEEEtran.bst: No hyphenation pattern has been}%
\typeout{** loaded for the language `#1'. Using the pattern for}%
\typeout{** the default language instead.}%
\else
\language=\csname l@#1\endcsname
\fi
#2}}
\providecommand{\BIBdecl}{\relax}
\BIBdecl

\bibitem{krs-ht-isit18}
K.~R. Sahasranand and H.~Tyagi, ``Extra samples can reduce the communication
  for independence testing,'' in \emph{2018 IEEE International Symposium on
  Information Theory (ISIT)}.\hskip 1em plus 0.5em minus 0.4em\relax IEEE,
  2018, pp. 2316--2320.

\bibitem{AhlCsi86}
R.~Ahlswede and I.~Csisz{\'a}r, ``Hypothesis testing with communication
  constraints,'' \emph{{IEEE} Trans. Inf. Theory}, vol.~32, no.~4, pp.
  533--542, July 1986.

\bibitem{Han87}
T.S.Han, ``Hypothesis testing with multiterminal data compression,'' vol.~33,
  no.~6, pp. 759--772, November 1987.

\bibitem{HanAmari98}
T.~S. Han and S.~Amari, ``Statistical inference under multiterminal data
  compression,'' \emph{{IEEE} Trans. Inf. Theory}, vol.~44, no.~6, pp.
  2300--2324, October 1998.

\bibitem{ShalabyPapamarcou92}
H.~M.~H. Shalaby and A.~Papamarcou, ``Multiterminal detection with zero-rate
  data compression,'' \emph{{IEEE} Trans. Inf. Theory}, vol.~38, no.~2, pp.
  254--267, March 1992.

\bibitem{XiangKimAllerton12}
Y.~Xiang and Y.~H. Kim, ``Interactive hypothesis testing with communication
  constraints,'' in \emph{50th Annual Allerton Conference on Communication,
  Control, and Computing}, October 2012, pp. 1065--1072.

\bibitem{XiangKimISIT13}
------, ``Interactive hypothesis testing against independence,'' in \emph{2013
  IEEE International Symposium on Information Theory}, July 2013, pp.
  2840--2844.

\bibitem{ZhaoLai14}
W.~Zhao and L.~Lai, ``Distributed testing against independence with multiple
  terminals,'' in \emph{52nd Annual Allerton Conference on Communication,
  Control, and Computing}, September 2014, pp. 1246--1251.

\bibitem{WiggerTimo16}
M.~Wigger and R.~Timo, ``Testing against independence with multiple decision
  centers,'' in \emph{International Conference on Signal Processing and
  Communications (SPCOM)}, June 2016, pp. 1--5.

\bibitem{katz2016collaborative}
G.~Katz, P.~Piantanida, and M.~Debbah, ``Collaborative distributed hypothesis
  testing with general hypotheses,'' in \emph{2016 IEEE International Symposium
  on Information Theory (ISIT)}.\hskip 1em plus 0.5em minus 0.4em\relax IEEE,
  2016, pp. 1705--1709.

\bibitem{sreekumar2019distributed}
S.~Sreekumar and D.~G{\"u}nd{\"u}z, ``Distributed hypothesis testing over
  discrete memoryless channels,'' \emph{IEEE Transactions on Information
  Theory}, 2019.

\bibitem{rahman2012optimality}
M.~S. Rahman and A.~B. Wagner, ``On the optimality of binning for distributed
  hypothesis testing,'' \emph{IEEE Transactions on Information Theory},
  vol.~58, no.~10, pp. 6282--6303, 2012.

\bibitem{hadar2019error}
U.~Hadar, J.~Liu, Y.~Polyanskiy, and O.~Shayevitz, ``Error exponents in
  distributed hypothesis testing of correlations,'' in \emph{2019 IEEE
  International Symposium on Information Theory (ISIT)}.\hskip 1em plus 0.5em
  minus 0.4em\relax IEEE, 2019, pp. 2674--2678.

\bibitem{hadar2019communication}
------, ``Communication complexity of estimating correlations,'' in
  \emph{Proceedings of the 51st Annual ACM SIGACT Symposium on Theory of
  Computing}, 2019, pp. 792--803.

\bibitem{andoni2019}
A.~Andoni, T.~Malkin, and N.~S. Nosatzki, ``Two party distribution testing:
  Communication and security,'' in \emph{46th International Colloquium on
  Automata, Languages, and Programming (ICALP 2019)}.\hskip 1em plus 0.5em
  minus 0.4em\relax Schloss Dagstuhl-Leibniz-Zentrum fuer Informatik, 2019.

\bibitem{dagan2018detecting}
Y.~Dagan and O.~Shamir, ``Detecting correlations with little memory and
  communication,'' in \emph{Conference On Learning Theory}, 2018, pp.
  1145--1198.

\bibitem{acharya2019inferenceI}
J.~Acharya, C.~L. Canonne, and H.~Tyagi, ``Inference under information
  constraints: Lower bounds from chi-square contraction,'' \emph{Proceedings of
  Machine Learning Research vol}, vol.~99, pp. 1--15, 2019.

\bibitem{acharya2019inferenceII}
------, ``Inference under information constraints ii: Communication constraints
  and shared randomness,'' \emph{arXiv preprint arXiv:1905.08302}, 2019.

\bibitem{pmlr-v89-acharya19b}
J.~Acharya, C.~Canonne, C.~Freitag, and H.~Tyagi, ``Test without trust: Optimal
  locally private distribution testing,'' in \emph{Proceedings of Machine
  Learning Research}, ser. Proceedings of Machine Learning Research,
  K.~Chaudhuri and M.~Sugiyama, Eds., vol.~89.\hskip 1em plus 0.5em minus
  0.4em\relax PMLR, 16--18 Apr 2019, pp. 2067--2076.

\bibitem{acharya2019domain}
J.~Acharya, C.~L. Canonne, Y.~Han, Z.~Sun, and H.~Tyagi, ``Domain compression
  and its application to randomness-optimal distributed goodness-of-fit,''
  \emph{arXiv preprint arXiv:1907.08743}, 2019.

\bibitem{hadar2019distributed}
U.~Hadar and O.~Shayevitz, ``Distributed estimation of gaussian correlations,''
  \emph{IEEE Transactions on Information Theory}, vol.~65, no.~9, pp.
  5323--5338, 2019.

\bibitem{Bonami70}
A.~Bonami, ``Etudes des coefficients {F}ourier des fonctiones de
  ${L}^p({G})$,'' \emph{Ann. Inst. Fourier}, vol.~20, no.~2, pp. 335--402,
  1970.

\bibitem{Gross75}
L.~Gross, ``Logarithmic sobolev inequalities,'' \emph{American Journal of
  Mathematics}, vol.~97, no.~4, pp. 1061--1083, 1975.

\bibitem{Beckner75}
W.~Beckner, ``Inequalities in {F}ourier analysis,'' \emph{Ann. of Math.}, vol.
  102, no.~1, pp. 159--182, July 1975.

\bibitem{Borell82}
C.~Borell, ``Positivity improving operators and hypercontractivity,''
  \emph{Mathematische Zeitschrift}, no. 180, pp. 225--234, 1982.

\bibitem{AhlswedeGacs76}
R.~Ahlswede and P.~Gacs, ``Spreading of sets in product spaces and
  hypercontraction of the markov operator,'' \emph{Ann. Probab.}, vol.~4,
  no.~6, pp. 925--939, December 1976.

\bibitem{Mossel13}
E.~Mossel, K.~Oleszkiewicz, and A.~Sen, ``On reverse hypercontractivity,''
  \emph{Geometric and Functional Analysis}, vol.~23, no.~3, pp. 1062--1097,
  June 2013.

\bibitem{GVJR2016}
V.~Guruswami and J.~Radhakrishnan, ``Tight bounds for communication-assisted
  agreement distillation,'' in \emph{Proceedings of the 31st Conference on
  Computational Complexity}, 2016, pp. 6:1--6:17.

\bibitem{AhlCsi98}
R.~Ahlswede and I.~Csisz{\'a}r, ``Common randomness in information theory and
  cryptography--part {II}: {CR} capacity,'' \emph{{IEEE} Trans. Inf. Theory},
  vol.~44, no.~1, pp. 225--240, January 1998.

\bibitem{TyaWat14}
H.~Tyagi and S.~Watanabe, ``A bound for multiparty secret key agreement and
  implications for a problem of secure computing,'' in \emph{EUROCRYPT}, 2014,
  pp. 369--386.

\bibitem{TyaWat14ii}
------, ``Converses for secret key agreement and secure computing,''
  \emph{{IEEE} Trans. Inf. Theory}, vol.~61, no.~9, pp. 4809--4827, September
  2015.

\bibitem{AcharyaDasKamath15}
J.~Acharya, C.~Daskalakis, and G.~Kamath, ``Optimal testing for properties of
  distributions,'' in \emph{Advances in Neural Information Processing Systems
  28}.\hskip 1em plus 0.5em minus 0.4em\relax Curran Associates, Inc., 2015,
  pp. 3591--3599.

\bibitem{LiuCuffVerdu17}
J.~Liu, P.~Cuff, and S.~Verd\'{u}, ``Secret key generation with limited
  interaction,'' \emph{{IEEE} Trans. Inf. Theory}, vol.~63, no.~11, pp.
  7358--7381, November 2017.

\bibitem{durrett2019probability}
R.~Durrett, \emph{Probability: theory and examples}.\hskip 1em plus 0.5em minus
  0.4em\relax Cambridge university press, 2019, vol.~49.

\bibitem{nair2014equivalent}
C.~Nair, ``Equivalent formulations of hypercontractivity using information
  measures,'' \emph{Proceedings of International Z{\"u}rich Seminar on
  Communications}, 2014.

\bibitem{CsiKor11}
I.~Csisz{\'a}r and J.~K{\"o}rner, \emph{Information theory: {C}oding theorems
  for discrete memoryless channels. 2nd edition}.\hskip 1em plus 0.5em minus
  0.4em\relax Cambridge University Press, 2011.

\end{thebibliography}


\end{document}